%% file: main.tex
\documentclass[a4paper]{article}
\usepackage[margin=2.5cm]{geometry}
\usepackage[]{amsthm} 
\usepackage{amsmath}
\usepackage{amssymb}
\usepackage{graphicx}
\usepackage{subcaption}
\usepackage{relsize}
\usepackage{wrapfig}
\usepackage{thmtools} 
\usepackage{enumitem}
\usepackage{xspace}
\usepackage{mdframed}
\usepackage{url}
\usepackage{bm}
\usepackage[numbers]{natbib}%
\definecolor{darkgreen}{rgb}{0, 0.5, 0}
\definecolor{darkred}{rgb}{0.5, 0, 0}
\usepackage[%
    colorlinks,
    citecolor=darkgreen,%
    linkcolor=darkred,%
    pdfborder={0 0 0}%
]{hyperref}
\usepackage{cleveref}

\newcommand{\Lang}{{L}}
\newcommand{\A}{\mathcal{A}}
\newcommand{\BigO}{{O}}
\newtheorem{theorem}{Theorem}
\newtheorem{definition}{Definition}
\newtheorem{hypothesis}{Hypothesis}

\newtheorem{lemma}{Lemma}
\newtheorem{corollary}{Corollary}
\newtheorem{observation}{Observation}
\theoremstyle{definition}
\newtheorem{example}{Example}
\newtheorem*{remark*}{Remark}

\title{Intersecting Dense Automata}

\author{Dmitry Chistikov\\ 
University of Warwick\\United Kingdom\\\texttt{d.chistikov@warwick.ac.uk} \and Neha Rino\\ 
University of Warwick\\United Kingdom\\\texttt{neha.rino@warwick.ac.uk}}
\date{}

\DeclareFontFamily{U}{wncy}{}
\DeclareFontShape{U}{wncy}{m}{n}{<->wncyr10}{}
\DeclareSymbolFont{mcy}{U}{wncy}{m}{n}
\DeclareMathSymbol{\Sh}{\mathord}{mcy}{"58}
\DeclareMathSymbol{\sh}{\mathord}{mcy}{"78}
\DeclareMathOperator*{\Shuffle}{\Sh} 
\newcommand{\shuffle}{\mathbin{\sh}}

\newcommand{\newextmathcommand}[2]{%
    \newcommand{#1}{\ensuremath{#2}\xspace}
}

\newcommand{\renewextmathcommand}[2]{%
    \renewcommand{#1}{\ensuremath{#2}\xspace}
}
\newextmathcommand{\nodprod}{\mathcal{A}_{\textit{nodding}}}
\newextmathcommand{\echoprod}{\mathcal{A}_{\textit{echoing}}}
\newextmathcommand{\catchprod}{\mathcal{A}_{\textit{catchup}}}
\newextmathcommand{\frogprod}{\mathcal{A}_{\textit{leapfrog}}}
\newextmathcommand{\dirpr}{\mathcal{A} \times \mathcal{B}}
\newcommand{\twoaut}{$\mathcal{A}$ and $\mathcal{B}$\xspace}
\newextmathcommand{\knfa}{\mathcal A_0, \ldots,  \mathcal A_{k-1}}
\newextmathcommand{\shortint}{\bigcap_{i\in [k]} \Lang(\A_i)}
\newextmathcommand{\longint}{L(\mathcal A_0) \cap \ldots \cap L(\mathcal A_{k-1})}
\newextmathcommand{\binint}{L(\mathcal{A}) \cap L(\mathcal{B})}
\newextmathcommand{\klangs}{\Lang_0, \ldots, \Lang_{k-1}}
\newextmathcommand{\prodqi}{\prod_{i\in [k]}Q_i}
\newextmathcommand{\qvec}{q_0, \dots, q_{k-1}}

\newcommand{\sset}{\subseteq}

\renewcommand{\emptyset}{\varnothing}
\renewextmathcommand{\epsilon}{\varepsilon}

\begin{document}

\maketitle

\begin{abstract}
\input{abstract}
\end{abstract}
\input{intro}

\input{prelim}
\input{constructions}
\input{deciding}
\input{certificates}
\input{extras}

\section*{Acknowledgements}
We are grateful to B\'eatrice B\'erard, Christoph Haase, Lia Sch{\"u}tze, and Michael Wehar for useful discussions, as well as to everyone who gave us feedback at various
stages of this work.

We would also like to thank
the Centre for Discrete Mathematics and its Applications (DIMAP) and
the Department of Computer Science, at the University of Warwick.
DC is supported by the Engineering and Physical Sciences Research Council [EP/X03027X/1].
NR is supported by
the Feuer International Scholarship in Artificial Intelligence by University of Warwick alumnus Jonathan Feuer.

\bibliographystyle{plainurl}
\bibliography{NFABib.bib}

\appendix
\input{cartesian}
\input{proof-shuffle-observation}
\input{leapfrogproof}
\input{relationsatisfactionproof}

\end{document}

%% file: abstract.tex
We observe that the classical Cartesian product construction for the
intersection of (languages of) nondeterministic finite automata (NFA) is
non-optimal in the worst case, if the automata have many transitions.
For a fixed alphabet, the product of two NFA may have $\Theta(m^2)$ transitions
if these NFA have at most $n$ states and $m$ transitions each.

We describe alternative constructions with $O(m n)$ transitions:
or $O(m n^{k-1})$ for the intersection of $k$ NFA
(for fixed $k \ge 2$ and alphabet~$\Sigma$).
This gives a faster algorithm for deciding NFA intersection emptiness.
The new algorithm is optimal, unless there exists a breakthrough
combinatorial algorithm for detecting $(k+1)$-cliques in undirected graphs.
This also leads to a more efficient certification scheme for NFA
intersection emptiness.

%% file: intro.tex
\section{Introduction}
\label{sec:intro}

The direct product construction for finite automata is a cornerstone of formal
language and automata theory~\cite{HMU}.
Given automata \twoaut with sets of states
$P$ and $Q$, respectively,
the product automaton \dirpr has as set of states
the Cartesian product $P \times Q$.
For each letter $a$ in the input alphabet~$\Sigma$, it contains
a transition $(p,q) \to (p',q')$ iff
$\mathcal{A}$ contains a transition $p \to p'$ and
$\mathcal{B}$ contains a transition $q \to q'$, both labelled by~$a$.
Thus, \dirpr simulates \twoaut simultaneously.
This construction shows
that the family of regular languages is closed under intersection:
$L(\dirpr) = L(\mathcal{A}) \cap L(\mathcal{B})$
for an appropriate choice of initial and final states.

The direct product of two deterministic finite automata (DFA) is deterministic,
but the construction applies to nondeterministic finite automata (NFA) too.
In fact, direct product constructions for more sophisticated forms of automata
are based on the same idea.

The number of transitions in the product automaton is what determines
the worst-case complexity of algorithms for
\emph{NFA intersection emptiness}, that is, those deciding
whether two given NFA accept a word in common.
In the model checking of finite-state systems, deciding NFA intersection emptiness
is a core sub-procedure in the verification of regular safety properties~\cite{KV01}.

It is well-known that the number of states in the automaton
\dirpr cannot, in general, be reduced:
that is, there exist finite automata \twoaut,
with $n$~states each, such that
every DFA (and even NFA) for the language \binint
has at least $n^2$ states (see, e.g.,~\cite{GruberHK-handbook}).
In this paper, we show that an analogous statement for the number of
transitions is false.
In other words, the direct product construction (for NFA) does \emph{not}
always produce automata with the smallest possible number of transitions.

Let NFA \twoaut have $n$ states each.
If they contain,
for some $a \in \Sigma$, many transitions labelled by~$a$ (are \emph{dense}),
then the number of transitions with this label in \dirpr may be as high as $n^4$.
We prove that, for any fixed input alphabet,
$O(n^3)$ transitions always suffice for a suitable NFA to recognise \binint,
with the number of states remaining $O(n^2)$ as in the direct product \dirpr.

The contributions of this paper are as follows.
We consider the intersection of languages recognised by
$k \ge 2$ NFA \knfa, each with at most $n$~states and $m$~transitions.\footnote{%
    We start numbering from $0$ instead of $1$.
    This will be convenient for our technical development.
}

\medskip



\textbf{1.}
We give several constructions of NFA
for the language \longint
that are sparser than the classical direct product NFA.
If $k \ge 2$ and the alphabet~$\Sigma$ are fixed,
our NFA have $O(n^k)$ states and $O(m n^{k-1}) = O(n^{k+1})$ transitions.
(We also give more precise estimates, tracking the dependence on $k$ and $|\Sigma|$
 without any assumptions about these parameters.)
This is an improvement in the number of transitions when compared with
the direct product, for which the tight upper bound is $m^k = O(n^{2 k})$.
The improvement is achieved at the cost of a constant-factor increase
in the number of states.
The idea of our construction is (for $k = 2$) to avoid the coupling of transitions of
the NFA $\mathcal A_0$ and $\mathcal A_1$ and instead \emph{interleave}
their runs.

\smallskip

\textbf{2.}
The construction of these sparser NFA immediately leads
to a faster algorithm to decide the emptiness of intersection \longint,
running in time $O(m n^{k-1}) = O(n^{k+1})$ if $k$ and $|\Sigma|$ are fixed.
Even for $k = 2$ this algorithm's running time, $O(m n)$,
improves over the running time of the naive algorithm based on the direct product,
$O(m^2)$ or $O(n^4)$.
For the latter algorithm,
dense NFA --- with many more transitions than states --- are the worst case.

\smallskip

\textbf{3.}
We show that our new algorithm is conditionally optimal, in the following sense.
For every $k \ge 3$,\footnote{%
    For $k = 2$, the reduction from triangle detection
    is already known~\cite{PotechinS20}; see Related Work below.
}
if there exists an algorithm which can determine, in time $O((m n^{k-1})^{1-\epsilon})$
for some constant $\epsilon > 0$, whether $\longint = \emptyset$ for NFA (or even DFA) \knfa
and which avoids fast matrix multiplication techniques (such as Strassen's), then
there also exists a breakthrough algorithm that detects the existence of $(k+1)$-cliques
in undirected graphs $G = (V, E)$ in time $O(|V|^{k - \delta})$ for some
constant $\delta > 0$
without using fast matrix multiplication.
No such algorithm is known, and in fact it has been conjectured it does not exist%
~\cite{AbboudBW15a}.

\smallskip

\textbf{4.}
We also show that intersection emptiness can be certified (and checked)
more efficiently, in the worst case, than by running the decision algorithm
from item~{2} above.
If $\longint = \emptyset$, there exists a certificate $C$ that witnesses
this fact, and, using fast matrix multiplication,
this certificate can be verified in time $O(n^{k + \omega - 2})$
for fixed $k$ and $|\Sigma|$. (Here $\omega \le 2.38$ is the infimum of all $c>0$
for which there exists an $O(n^c)$-time algorithm to multiply two $n \times n$ matrices.)
Thus, the verification time beats the worst-case estimate $O(n^{k+1})$ from
item~{2} for dense enough NFA.

\medskip

Our most elementary construction of sparser NFA for \longint
(dubbed \emph{nodding product})
relies on $\epsilon$-transitions and
has at most $k |\Sigma| \cdot  n^k$ states and $k \cdot mn^{k-1}$ transitions.
The NFA intersection emptiness algorithm is a search in its transition graph.
If we require an NFA (not an $\epsilon$-NFA) for the intersection language,
our constructions
(\emph{catch-up product} and \emph{leapfrog product})
suffer a factor of $|\Sigma|^{k}$ (or $|\Sigma|^{k-1}$;
in any case constant if $k$ and $\Sigma$ are fixed) growth
in the number of states and transitions compared to the nodding product.
The resulting NFA is still smaller than the direct product as long as $|\Sigma| = o(n)$.

\subsection*{Related work}

\paragraph*{Constructions of sparser NFA; faster algorithm.}
The main ingredient of our first two contributions
is well-known:
the interleaving of equal-length words
where letters alternate~\cite[Exercise~4.2.7]{HMU}
has been studied under several names, `perfect shuffle'
and `literal shuffle' among them~\cite{Berardshuffle,Berardformal,HenshallRS12,Hoffmann25}.

Nevertheless, we have been unable to find these results in the literature.
The closest is the work of
Boigelot and Wolper~\cite{BoigelotWolper},
who interleave transitions (rather than fire them simultaneously)
for automata that encode sets of satisfying assignments to formulas
of linear arithmetic.
This reduces the size of description of the product automaton
in this application.
However,
both the constructions of sparser NFA in general and the faster algorithm
for NFA intersection emptiness appear to be new. 
For example, a recent textbook by Esparza and Blondin
gives the direct product--based algorithm for the intersection
of two NFA, with running time bounded by the product of the number of
transitions in them~\cite[Sec.~3.2.3]{EsparzaBlondin}.

Baier and Katoen's monograph on model checking~\cite[Chap.~4]{BK}
also reduces the NFA intersection problem to a reachability check
in the direct product.
For the verification of regular safety properties,
Baier and Katoen
(following Kupferman and Vardi~\cite{KV01}) describe
how to apply the direct product construction
to a finite labelled transition system
and an NFA that recognises
the language of `bad' prefixes for the property.
When transitions of the LTS are re-labelled with
subsets of atomic propositions, rather than with actions,
the usual product construction on NFA can be applied;
see Vardi and Wolper~\cite{VW-LICS86}.
The complexity of verification algorithm is bounded from
above by the product of the number of transitions
in the system and the NFA~\cite[Thm.~4.22]{BK}.
Our new algorithm would in this case traverse a graph
with fewer transitions than this pessimistic estimate,
reducing the complexity.

\paragraph*{Fine-grained complexity: lower bounds on the computational complexity.}
Optimising the degree of polynomials, $c$,
in (polynomial) running time bounds of the form $O(n^c)$, where
$n$ is the size of the input to an algorithmic problem, is a goal of
\emph{fine-grained complexity} research; see, e.g.,~\cite{WilliamsFG2018,Bringmann19}.
This field complements faster algorithms by lower bounds,
albeit conditional.

The first such bounds for the intersection emptiness problem
for $k$ DFA (\emph{DFA~$k$-IE})
were obtained by Kasai and Iwata~\cite{KasaiIwata},
who observed that
the intersection of $k$~DFA can simulate a nondeterministic $k \log n$ space--bounded Turing machine;
see also an earlier paper by Kozen~\cite{Kozen}.
Karakostas, Lipton, and Viglas~\cite{KarakostasLV03} prove
that if there exists an $n^{o(k)}$-time algorithm for DFA $k$-IE,
then there are also
deterministic algorithms for solving subset sum and factoring $n$-bit integers in
$2^{\epsilon n}$ time for every $\epsilon > 0$.
This implication can be considered a conditional lower bound
on the complexity of DFA $k$-IE.
Wehar~\cite{Wehar14} later showed the
existence of $f(k) = o(k)$ such that, if DFA $k$-IE has
an $n^{f(k)}$ algorithm, then $\textsf{NL} \ne \textsf{P}$.
A different angle on the problem was taken by Wareham~\cite{Wareham00}, who
characterised the parameterised complexity of DFA $k$-IE
for words of specified length, and related problems.

Connections have also been discovered between DFA $k$-IE
and Boolean satisfiability, SAT, in particular using the prominent
Exponential Time Hypothesis (ETH) and Strong
Exponential Time Hypothesis (SETH;
 see~\cite{IP01} and, e.g.,~\cite{WilliamsFG2018,Bringmann19}).
Fernau and Krebs~\cite{DBLP:journals/algorithms/FernauK17} proved
that an $n^{o(k)}$ algorithm for DFA $k$-IE would contradict ETH.
Wehar~\cite[Cor.~7.19]{wehar2017complexity} proved that,
for every fixed $k$, an $\BigO(n^{k - \epsilon})$-time algorithm for DFA $k$-IE would contradict
SETH~\cite{IP01}.
We discuss connections with SETH and its stronger nondeterministic variant
NSETH~\cite{CarmosinoGIMPS16} in \Cref{sec:certificates}.

For the cases $k = 2$ and $k = 3$, de Oliveira Oliveira and Wehar~\cite{OliveiraW18}
show reductions between DFA $k$-IE, triangle detection in undirected graphs,
and the 3-SUM problem.
Potechin and Shallit~\cite{PotechinS20} reduce
triangle detection to the special case of intersection emptiness of two NFA:
an apparently easier problem (\emph{NFA acceptance}) of deciding whether
a given input word $w\in \Sigma^*$ is accepted by a given NFA~$\mathcal A$.
This reduction works even with the unary alphabet ($|\Sigma| = 1$).

In contrast with a substantial body of work on DFA~$k$-IE,
the natural analogous problem for NFA has apparently not received
much attention. As our results show, the complexity landscape
for NFA $k$-IE turns out to be different.

The usage of fast matrix multiplication as part of the verifier algorithm in
our certification scheme raises the question of whether this technique
could also speed up the decision procedures for NFA intersection emptiness.
This question appears to be rather difficult.
Even for the NFA acceptance problem
it is unknown if fast matrix multiplication techniques could be used to
obtain an algorithm with worst-case running time $O((|w| \cdot m)^{1-\epsilon})$
for some $\epsilon > 0$.
Here $w$ is the length of the input word and $m$ the number of transitions
in the NFA.
Bringmann et~al.~\cite{BringmannGKL24} recently conjectured that
no such algorithm exists.

%% file: prelim.tex
\section{NFA intersection}
\label{sec:prelim}
	\begin{definition}[NFA]
		A nondeterministic finite automaton $\A$ is a tuple $(Q, \Sigma, \Delta, s , F)$ consisting of a finite nonempty set of states $Q$, a finite input alphabet $\Sigma$, an initial state $s \in Q$, a set of final states $F \subseteq Q$, and a transition relation $\Delta \subseteq Q \times \Sigma \times  Q$. 
	\end{definition}

A \emph{deterministic finite automaton} (DFA) is an NFA whose transition relation is in fact a partial function $\delta \colon Q \times \Sigma \to Q$. 

Every NFA is associated with a labelled graph called its transition graph. 
The vertices of the transition graph are $Q$, edge labels belong to $\Sigma$, and the labelled edge relation is given by $\Delta$.

A \emph{run} $\rho$ of an NFA $\A$ on a word $w \in \Sigma^*$ is a string $t_0t_1\dots t_r \in \Delta^*$, i.e., a string of transitions with the following properties: 
\begin{enumerate}
    \item It starts at the initial state, i.e., $t_0 \in \{s\}\times \Sigma \times Q$. 
    \item Let each $t_i = (q_i, \sigma, q'_i)$. The unlabelled versions of the transitions $(q_0, q'_0)(q_1, q'_1)\dots (q_r, q'_r)$ form a path, i.e., for all $i \in \{0, \ldots, r-1\}$, $q'_i = q_{i+1}$ . 
    \item The concatenation of the labels of the transitions form $w$, i.e.,  $\sigma_0\sigma_1\dots \sigma_r = w$. 
\end{enumerate}

We assume that an empty run (corresponding to the empty word in $\Delta^*$) starts and ends at the initial state. 
The destination of the last transition, $q'_r$, is said to be \emph{reached} (from the initial state~$s$).
The term \emph{accessible part} (of an automaton) refers to the states that can be reached from the initial state, and to their outgoing transitions. 

 An \emph{accepting} run is a run that ends in a final state. 
 The NFA $\A$ \emph{accepts} a word $w$ if there is an accepting run of $\A$ on $w$. 
The \emph{language} of an NFA, $\Lang(\A)$, is $\{w \in \Sigma^* \mid \text{$\A$ has an accepting run on $w$}\}$.

We assume that $\Delta$ is represented as a list of adjacency lists. 
Given an enumeration of the states $Q = \{q_0, \dots, q_{n-1}\}$, for each $\sigma \in \Sigma$, the adjacency matrix $\Delta(\sigma)$ is an $n \times n$ matrix such that $\Delta(\sigma)[i,j]$ is $ 1 $ if $(q_i, \sigma, q_j) \in \Delta$, and $\Delta(\sigma)[i,j]$ is $0$ otherwise. 

	\begin{definition}[$\varepsilon$-NFA]
		An $\varepsilon$-nondeterministic finite automaton $\A$ is a tuple $(Q, \Sigma, \Delta, s , F)$ consisting of a finite nonempty set of states $Q$, a finite input alphabet $\Sigma$, an initial state $s \in Q$, a set of final states $F \subseteq Q$, and a transition relation $\Delta \subseteq Q \times (\Sigma \cup \{\varepsilon\}) \times  Q$, where $\varepsilon$ is the empty word. 
	\end{definition}

Runs, accepting runs, and the language accepted by an $\varepsilon$-NFA are defined analogously. 

\begin{definition}[NFA $k$-IE]
    For a fixed $k$, given $k$ NFA $\knfa$ with at most $n$ states and $m$ transitions each, over an $\ell$-letter alphabet, the NFA $k$-Intersection Emptiness (NFA $k$-IE) problem asks,  is $\shortint = \emptyset$?
\end{definition}

\begin{mdframed}
Notation ($k$, $\ell$, $m$, $n$).\\
Throughout the paper, we consider the intersection
of languages of $k$~automata \knfa, each with:
\begin{itemize}
\item input alphabet $\Sigma$ of size $\ell$,
\item at most $m$ transitions, and
\item at most $n$ states.
\end{itemize}
For simplicity, \textbf{we assume} $m \ge n$ (rather than $m \ge n-1$,
which holds if all states in the transition graph of each $\A_i$
can be reached from the initial state).
\end{mdframed}

\smallskip

When discussing the complexity of our constructions and algorithms,
we give `uniform' bounds, with parameters $m$ and $n$ across all
the NFA \knfa.
An interested reader will easily refine these bounds
for the scenario where some NFA are much bigger than the other ones.

The direct product construction for NFA is recalled in \Cref{sec:CPA}.

%% file: constructions.tex
\section{Sparser automata for \mbox{NFA intersection}}\label{sec:sparser}

In this section we describe new product constructions
for NFA. These can be sparser than the standard direct product;
this is in particular true if $k, \ell = O(1)$.

For $k \geq 2$, we denote $[k]= \{0, \dots , k-1\}$. 

To begin with, we recall a few basic operations on
words and languages.
Let $\Sigma$ be a finite alphabet.
    Given a word $w = x_0 \dots x_{t-1} $  where $x_j \in \Sigma$, and nonnegative integers $i, k$ with $i < k$,
    the \emph{$i \bmod k$ restriction} is the word
    \begin{equation*}
    w{\upharpoonright}_{i \bmod k} = x_i x_{i+k} \dots x_{i + kt'}
    \end{equation*}
    where  $t'$ is the largest integer such that $i + k t' < t$.
    For languages, $\Lang{\upharpoonright}_{{i \bmod k }}  = \{w{\upharpoonright}_{i \bmod k} \mid w \in \Lang\}$.

    The \emph{interleaving of two languages} $\Lang_0, \Lang_1 \subseteq \Sigma^*$ is
    \begin{align*}
    \Lang_0 \shuffle \Lang_1 =
    \bigl\{x_0y_0\dots x_{t-1} y_{t-1} \mid\,  x_0 \dots x_{t-1} \in \Lang_0 \text{ , } y_0 \dots y_{t-1} \in \Lang_1  \text{ and } x_i, y_i \in \Sigma \text{ for all }i\in [t] \bigr\}.
    \end{align*}
    That is, all pairs of words from $\Lang_0$ and $\Lang_1$ are interleaved letter by letter (like a perfectly riffle shuffled deck of cards).
    For $k > 2$, given languages $\Lang_0, \Lang_1, \dots , \Lang_{k-1} \subseteq \Sigma^*$, their \emph{interleaving}, denoted by \raisebox{0pt}[1ex][0pt]{$\Shuffle\limits_{i\in [k]} \Lang_i$}, is the language
    \begin{align*}
    \bigl\{w \in \Sigma^* \mid\, & \text{$w{\upharpoonright}_{i \bmod k} \in \Lang_{i}$ for all $i \in [k]$ and $|w| \equiv 0 \bmod k$}\bigr\}.
    \end{align*}
As is well-known, if $\klangs$ are regular
and $i \in [k]$, then languages
\raisebox{0pt}[1ex][0pt]{$ \Shuffle\limits_{i\in [k]} \Lang_i $} and $\Lang_0{\upharpoonright}_{{i \bmod k }}$ are regular too.

\begin{observation}\label{obs:fact} 
Let $\klangs \sset \{a_0, \ldots, a_{\ell-1}\}^*$. Then
\begin{equation*}
\bigcap_{i \in [k]} L_i = [(\Shuffle_{i\in [k]} \Lang_i ) \cap (a_0^k + \dots + a_{\ell-1}^k)^*]\mathlarger{{\upharpoonright}}_{0 \bmod k}\;.
\end{equation*}
\end{observation}

We will not rely on \Cref{obs:fact} or the definition of $\Shuffle$ explicitly, but
they are at the core of our constructions.
For completeness, the proof of \Cref{obs:fact} is given
in Appendix~\ref{sec:interleavingproof}.

Throughout the rest of this section we assume that parameters $n, m, k, \ell$
and NFA \knfa are defined as in Section~\ref{sec:prelim}.

\subsection{Nodding product automaton}
\label{sec:nod}
\paragraph{Idea of the construction.}
The nodding product automaton $\nodprod$ is an  $\varepsilon$-NFA that recognises $\shortint$, while being sparser than the direct product automaton. 
The idea is to use \Cref{obs:fact}.
Take the standard NFA for the interleaving of $k$ languages, then intersect it with $ (a_0^k + a_1^k + \dots + a_{\ell-1}^k)^*$, and finally apply the $0 \bmod k$ restriction. 
All three operations preserve regularity, yielding an NFA for the intersection. 

Intuitively, in rounds, the first automaton $\A_0$ announces the next letter, and the other automata silently nod (read $\varepsilon$) in turn.

The benefit of the construction is that the NFA for the language $\Shuffle_{i\in [k]} \Lang(\A_i) $
is much sparser than the direct product NFA $\A_0 \times \ldots \times \A_{k-1}$.
In the interleaving, the `component automata' move one at a time, and therefore
only $\BigO(k \cdot mn^{k-1})$ transitions are needed. 
The resulting automaton is $k$-partite, so the $0 \bmod k$ restriction does not require more states, as long as $\varepsilon$-transitions can be used. 
Overall, this approach yields a $\BigO(k \ell \cdot m n^{k-1})$ sized $\varepsilon$-NFA. 
Our construction below slightly improves this bound, removing the $\ell$ factor.

\paragraph{Definition of $\nodprod$.} 
Consider first the special case of
two NFA $\A_0$ and $\A_1$ over the alphabet $\Sigma = \{a, b, c\}$.
We define
$$ \nodprod = (Q, \Sigma, \Delta, s , F ).$$ 
Here the states are $Q = Q_0 \times Q_1 \times (\Sigma \cup \{\varepsilon\})$.
The initial state is $s =(s_0, s_1, \varepsilon)$ and the set of final states is $F= F_0 \times F_1 \times \{\varepsilon\}$. 
The set $\Delta$ consists of the following transitions:
\begin{itemize}
\item
$(p_0, q, \epsilon) \xrightarrow[]{\sigma} (p_1, q, \sigma) $ whenever  $p_0 \xrightarrow[]{\sigma} p_1$ in $\A_0$,  
\item
$(p, q_0, \sigma) \xrightarrow[]{\varepsilon} (p, q_1, \epsilon) $  whenever  $q_0 \xrightarrow[]{\sigma} q_1$ in $\A_1$. 
\end{itemize}
Here $\sigma$ ranges over $\Sigma$.

In words, $\nodprod$ has a copy of $Q_0 \times Q_1$ for each letter in $\{a, b, c\}$ as well as a \emph{base} copy of $Q_0 \times Q_1$, labelled by $\varepsilon$.  
Given two copies of $Q_0 \times Q_1$, we refer to the set of all transitions from the first copy to the second as a \emph{volley} (as in `a volley of arrows').
The four copies of $Q_0\times Q_1$ are arranged in three different $2$-cycles, or petals, where each petal consists of the base copy and a letter-labelled copy, connected in either direction by volleys; see \Cref{fig:nodding}. 
The petals are glued at the common base copy, like a flower. 
For each  $\sigma \in \Sigma$,  the volley from the base copy to the $\sigma$-copy  reads $\sigma$ and updates the $Q_0$ component of the state as if $\A_0$  read $\sigma$. 
The volley from the $\sigma$-copy back to the base copy reads $\varepsilon$ while updating the $Q_1$ component as if $\A_1$ read $\sigma$. 
The structure of $\nodprod$ mimics that of the NFA $\mathcal{C}$ that recognises the language $(aa + bb + cc)^*$; see \Cref{fig:aabbcc}. 

\begin{figure*}

    \centering
    \begin{subfigure}[t]{0.25\textwidth}
        \centering
        \includegraphics[width=0.9\textwidth]{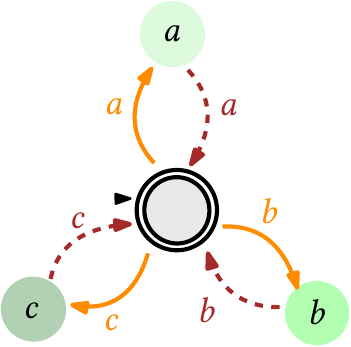} 
\caption{\rlap{NFA for $(aa + bb + cc)^*$}\phantom{NFA for $0123456$}}
    \label{fig:aabbcc}    
    \end{subfigure}%
    ~ 
    \begin{subfigure}[t]{0.65\textwidth}
        \centering
        \includegraphics[width=0.9\textwidth]{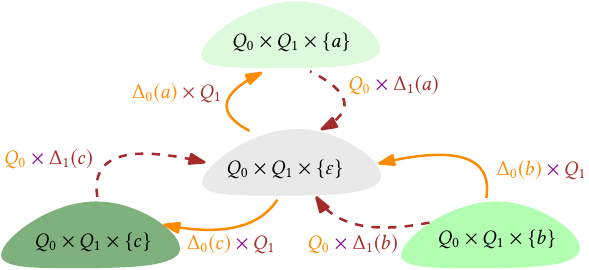} 
\caption{NFA $\nodprod$}    \label{fig:nodding}    
    \end{subfigure}

     \caption{Nodding product automaton}

\end{figure*}

In general, given $k \ge 2$ NFA $\knfa$ over an $\ell$-letter alphabet $\Sigma$, the $\varepsilon$-NFA $\nodprod$ contains $(k-1)\ell$ copies of the direct product $ \prodqi$, and a base copy.
These copies are arranged in $\ell$ different $k$-cycles or petals which share the base copy. 
We number volleys from $0$ to $k-1$. 
The zeroth volley in each petal is the outgoing volley from the base copy, proceeding in order till the $(k-1)$th volley which is the incoming volley to the base copy in this petal. 
For each  $\sigma \in \Sigma$, in the $\sigma$-petal the first volley reads $\sigma$ and updates the $Q_0$ component as if $\A_0$ read $\sigma$. 
For each $i \in [1,k-1]$, the $i$th volley in the $\sigma$-petal reads the empty string $\varepsilon$ while updating the $Q_{i-1}$ component as if $\A_{k-1}$ read $\sigma$. 

\begin{theorem}\label{thm:nodding}
    The $\varepsilon$-NFA $\nodprod$: \begin{itemize}
        \item recognises $\shortint$,
        \item has at most $(k \ell - \ell + 1)\cdot  n^k$ states and at most $k \cdot mn^{k-1}$ transitions, 
        \item  and has accessible part which can be constructed in $\BigO(k \cdot mn^{k-1})$ time.
    \end{itemize}
\end{theorem}

\begin{proof}[Proof (sketch)]
    Every accepting run of $\nodprod$ starts at the initial state in the base copy, traverses a whole number of petals, and ends at a final state in the base copy. 
    A run of $\nodprod$ on a word $w = \sigma_0\dots \sigma_{t-1}$ traverses the $\sigma_0$-petal, returns to the base, then traverses the $\sigma_1$-petal and so on. 
    Any such run is the interleaving of $k$ runs, one in each component automaton; see \Cref{fig:noddingrun}. 
    (The first is a run of $\A_0$ on $w$, and the others are runs of $\A_i$ on $w$, $i \in [1, k-1]$,  with edge labels being $\varepsilon$ throughout rather than $\sigma_0, \dots, \sigma_{n-1}$.)
    In summary, $w$ is accepted if and only if it is accepted by all the component automata, so $\nodprod$  recognises $\shortint$. 

    The bound on the number of states follows from the inequalities $|Q_i| \le n$.
    Denote by $\Delta_i(\sigma)$ the set of all $\sigma$-transitions in NFA~$\A_i$, for $\sigma \in \Sigma$.
    Then
    \begin{equation*}
    \qquad
    |\Delta| =
    \sum_{i \in [k]} \sum_{\sigma \in \Sigma}\left( \prod_{j \neq i} |Q_j|\right)|\Delta_i(\sigma)|
        \le k \cdot n^{k-1} \sum_{\sigma \in \Sigma}\hspace{0.2em} \max_{i \in [k]}\hspace{0.2em}{|\Delta_i(\sigma)|} \le k \cdot m n^{k-1}.
    \end{equation*}
    The number of accessible states in an NFA cannot exceed the number of transitions plus one.
    Thus, the accessible part of $\nodprod$ can be constructed in $\BigO(k \cdot m n^{k-1})$ time.
\end{proof}
\begin{figure*}
        \centering
        \includegraphics[width=0.42\textwidth]{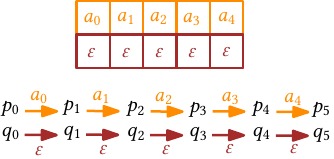} 
\caption{A run of $\nodprod$}    \label{fig:noddingrun}    
    \end{figure*}
\Cref{fig:noddingrun} shows a representation of a run of $\nodprod$ as a composition of runs of $\A_0$ and $\A_1$, using little squares or \emph{bricks}. Each brick represents a transition in $\nodprod$. 
The top orange row of bricks reflects how the run in the $Q_1$ component advances (like a progress bar), and the bottom brown row of bricks reflects the $Q_2$ component. The label $\sigma \in \Sigma \cup \{\varepsilon\}$ inside a brick signifies what is being read by the transition. 

An NFA very similar to \nodprod, but without $\varepsilon$-transitions, can be constructed
for the language of all $k$-stutterings of words in \shortint.
The \emph{$k$-stuttering} of $w = \sigma_0 \ldots \sigma_{t-1} \in \Sigma^t$ simply
repeats each letter $k$ times, in the same order: $\sigma_0^k \ldots \sigma_{t-1}^k$.
This NFA, which we call the \emph{echoing product} automaton, \echoprod,
is the same as \nodprod but
with $\varepsilon$-transitions re-labelled with elements of $\Sigma$ appropriately:
its construction recalls the same \Cref{obs:fact}
but does not apply the final $0 \bmod k$ restriction.
Intuitively, in rounds, the first automaton $\A_0$ announces the next letter, and the other automata echo this letter.

\subsection{Catch-up product automaton}
\begin{figure*}[t!]
    \centering

    \begin{subfigure}[t]{0.8\textwidth}
        \centering
        \includegraphics[width=0.99\textwidth]{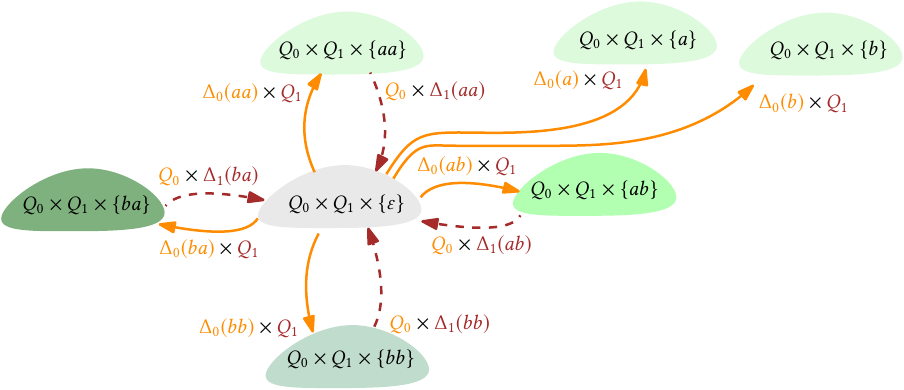} 
\caption{NFA $\catchprod$}    \label{fig:catch-up}    
    \end{subfigure}
~ \hspace{-18ex}
        \begin{subfigure}[t]{0.25\textwidth}
        \centering
        \includegraphics[width=0.98\textwidth]{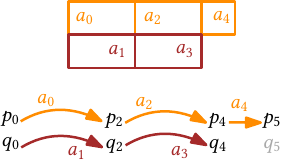} 
\caption{A run of $\catchprod$}    \label{fig:catch-uprun}    
    \end{subfigure}
    \caption{Catchup product automaton}
\end{figure*}
\paragraph{Idea of the construction.}
The catch-up product automaton is an NFA (not an $\varepsilon$-NFA) that recognises $\shortint$ and is sparser than the direct product automaton.
It is bigger than $\nodprod$ but does not have $\varepsilon$-transitions. 

Like $\nodprod$, the catch-up product automaton updates component automata one at a time.
Take $k = 2$.
Simulating two runs on input $w$ amounts to $2 |w|$ moves.
Our new NFA $\catchprod$ avoids $\varepsilon$-transitions and thus
performs two moves (either of $\A_0$ or of $\A_1$) per transition.
This is shown in the brick diagram in \Cref{fig:catch-uprun}.
The length of a brick signifies the number of moves by which the component run advances during that transition.  
In the figure, $\catchprod$ reads $a_0$ and updates the $Q_0$ component by two moves (thus committing to $a_1$ as well);
then reads $a_1$ and advances the $Q_1$ component by two moves;
then reads $a_2$ and advances the $Q_0$ component by two moves (committing to $a_3$), etc.

Intuitively, in rounds, the first automaton moves forward $k$ steps and then waits as the rest catch up to it, one by one.  

\paragraph{Definition of $\catchprod$.}
First consider $k = 2$.
We define the catch-up product automaton $$\catchprod = (Q, \Sigma, \Delta,s, F).$$ 
Here $Q = Q_0 \times Q_1\times (\{\varepsilon\} \cup \Sigma \cup \Sigma^2)$ and $s = (s_0, s_1, \varepsilon)$. 
The set $\Delta$ includes transitions:
\begin{itemize}
\item
$(p_0, q, \epsilon) \xrightarrow[]{\sigma_0} (p_2, q,  \sigma_0\sigma_1)$ whenever $p_0 \xrightarrow[]{\sigma_0} p_1 \xrightarrow[]{\sigma_1}p_2$ is a path in $\A_0$,
\item
$ (p, q_0, \sigma_0\sigma_1) \xrightarrow[]{\sigma_1} (p, q_2, \epsilon) $ whenever $q_0 \xrightarrow[]{\sigma_0} q_1 \xrightarrow[]{\sigma_1}q_2$ is a path in $\A_1$
(we note that pre-computation of the $2$-step reachability relations in $\A_0$ and $\A_1$ is necessary),
\item
$(p_0,q, \epsilon) \xrightarrow[]{\sigma} (p_1,q,\sigma)$ whenever $p_0 \xrightarrow[]{\sigma} p_1$ is a transition in $\A_0$ (for each $q \in Q_1$ and $\sigma \in \Sigma$). 
\end{itemize}
There are no outgoing transitions from the $\Sigma$-labelled copies. 
The set of final states $F$ is $(F_0 \times F_1 \times \{\varepsilon\}) \cup \{(f_0, q_1, \sigma) \mid (q_1, \sigma, f_1) \in \Delta_1 \text{ and } (f_0, f_1) \in F_0 \times F_1 \}$.

Given $k$ NFA, $\catchprod$ is a flower with a central base copy and $\ell^k$ petals, one for each $k$-letter word over $\Sigma$  (see \Cref{fig:catch-up} for the case $\Sigma = \{a, b\}$).
Each petal is a $k$-cycle of copies of the Cartesian product \prodqi, connected by $k$ volleys. 
Given $ u = \sigma_0\dots \sigma_{k-1} \in \Sigma^k$, in the $u$-petal, starting from the base copy, the first volley reads $\sigma_0$ and updates the $Q_0$ component as if $\A_0$ read $u$. The next volley reads $\sigma_1$ and updates the $Q_1$ component as if $\A_1$ read $u$, and so on till the last volley, which reads $\sigma_{k-1}$, updates the $Q_{k-1}$ component, and returns to the base copy. 

There are also trailing paths, or \emph{tails}, from each state in the base copy, for each word $v \in \Sigma^{<k}$, to accept words whose length is not divisible by $k$.  For a given $v = \sigma_0\dots \sigma_{t-1}$, the $v$-tail consists of $t$ copies of \prodqi, connected by volleys. The first volley in the $v$-tail reads $\sigma_0$ and updates the $Q_0$ component as if $\A_0$ read $v$; the second volley reads $\sigma_1$ and updates the $Q_1$ component as if $\A_1$ read $v$; etc. At the end of the tail, in the $t$th copy of \prodqi, the first $t$ components are up to date, but the rest are not. In this tail-end copy, a state $(q_1, \dots, q_k, v, t)$ is accepting if the first $t$ components of the state are final states, and for every $i$ in $[t,  k-1]$, $q_i \xrightarrow[]{v} f_i$ in $\A_i$ where $f_i \in F_i$. 

The initial state of $\catchprod$ is the initial vector in the base copy, and the final states are the product of the final states in the base copy, as well the ends of the tails described previously. 

Given a word $u$ of length at most $k$, we define the \emph{$u$-reachability relation} of $\A_i$, denoted by $\Delta_i^{u}$, to be the set of pairs of states $(p, q) \in Q_i^2$ such that there is a path in $\A_i$ from $p$ to $q$ reading $u$. 
Denote $m_{\leq k} = \max\limits_{i < k} \max\limits_{u \in \Sigma^{\leq k}} |\Delta_i^{u}|$. Note that $m_{\leq k} \le n^2$, because each $\Delta_i^{u}$ fixes the word $u$, irrespective of $\ell$.

\begin{theorem}\label{thm:catch-up}
    The NFA $\catchprod$: \begin{itemize}
\item recognises $\shortint$,
        \item has at most $2k\ell^{k} \cdot n^k $ states and at most $2k\ell^{k} \cdot m_{\scriptscriptstyle{\leq k}} \hspace{0.25em}n^{k-1} $ transitions, and
        \item has accessible part which can be constructed in $\BigO(k \ell^{k} \cdot ( m_{\scriptscriptstyle{\leq k}} \hspace{0.25em}n^{k-1} + k \cdot n^{\omega}))$ time.
    \end{itemize}  
\end{theorem}
\begin{proof}[Proof (Sketch)]
     Suppose $\catchprod$ reads a word $w$, where $k$ divides $|w|$. 
    The run, call it $\pi$, starts at the initial state in the base copy, traverses a whole number of petals, and returns to the base copy. 
 
    Consider the restriction $\pi {\upharpoonright}_{0 \bmod k}$:
    in it, only the first volley in each petal is taken.
    Let $t = |w|/k$.
    In the sequence of transitions $\pi {\upharpoonright}_{0 \bmod k}$,
    the $Q_0$ component of $\catchprod$ jumps $k$ moves at a time:
    $q_0 \xrightarrow[]{w_0} q_k \xrightarrow[]{w_1} q_{2k} \xrightarrow[]{w_2} \dots \xrightarrow[]{w_{t-1}} q_{|w|}$,
    where $w = w_0 \ldots w_{t-1}$ and $|w_j|=k$ for all~$j$.
    This yields a run of $\A_0$ on $w$.
    Considering other restrictions similarly yields runs of the other $\A_i$'s on $w$.
    Hence, a run reaches a final state in the base copy if and only if all $k$ of the corresponding component runs were also accepting. 
    
    When $|w| \neq 0 \mod k$, the run proceeds in the petals as long as possible, reaching a state in the base copy, e.g., $(\qvec)$. At this point there are only $|w| \bmod k$ letters left to be read, say $v \in \Sigma^+$. The $(\qvec, v)$ tail is the only one that can read $v$. Such a run accepts $w$ if and only if for each component $\A_i$ there is a run on $w$ that starts at $s_i$, reaches $q_i$ with only $v$ left to be read, and $q_i$ could read $v$ to reach a final state, i.e., if $\A_i$ accepts $w$.    
    \begin{figure*}[t!]
    \centering
     \begin{subfigure}[t]{0.65\textwidth}
         \centering
         \includegraphics[width=0.99\textwidth]{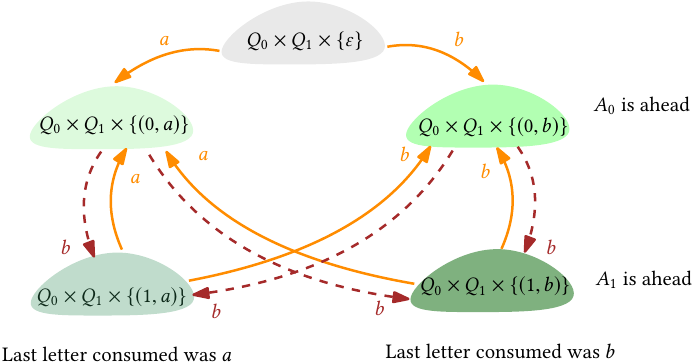} 
     \end{subfigure}

        \caption{Leapfrog product automaton $\frogprod$}
    \label{fig:leapfrog}    

\end{figure*}
    Hence, $\catchprod$ recognises \shortint.  
    
     In $\catchprod$ there are at most $k$ copies of \prodqi for every word of length at most $k$, so overall there are at most $2k\ell^{k} \cdot n^k$ states. Moreover, for each $i \in [k]$, for every word $u$ of length at most $k$, for every element of the $u$-reachability relation $\Delta_i^{u}$, and for every $(k-1)$-tuple of states from components other than $Q_i$, we have at most one transition of $\catchprod$. 
    Hence, $|\Delta| \leq 2k\ell^{k} \cdot m_{\scriptscriptstyle{\leq k}} \hspace{0.25em}n^{k-1} $. 

    The accessible states and transitions can be constructed as needed from the initial state, and their number is dominated by the number of transitions. 
    
    In order to construct the transitions of $\catchprod$, we need to precompute all of the $u$-reachability relations, i.e., determine if $q \xrightarrow[]{\sigma_0\ldots \sigma_{t-1}} q'$ for each automaton $\A_i$ and each word $u = \sigma_0 \ldots \sigma_{t-1} 
    \in \Sigma^{\leq k}$, $t \leq k$. For every such word, this relation is obtained multiplying $t$ adjacency matrices appropriately.

    Overall, this precomputation takes time $\BigO(k^2 \ell^k \cdot n^{\omega})$.  
    The upper bound on the total running time follows.
\end{proof}

\subsection{Leapfrog product automaton}

To construct $\catchprod$, we have pre-computed the $k$-step reachability relations. Given these relations, we can reduce the number of states in the product automaton even further. 
For example, when processing the word $a_0a_1a_2a_3a_4$ in two automata, the product automaton could read $a_0$ and update  the $Q_0$ component as if $\A_0$ read $a_0$, then read $a_1$ and update the $Q_1$ component as if $\A_1$ read $a_0a_1$, then read $a_2$ and update the $Q_0$ component as if $\A_0$ read $a_1a_2$, and so on (see \Cref{fig:leapfrogrun}). 
By staggering the bricks as shown, to perform such updates we only need to store $\Sigma$ in memory. 

Intuitively, each NFA, in turn, leaps over the others to become the furthest ahead.

As an example, for two NFA, for each  $\sigma_0\sigma_1 \in \Sigma^2$, the set of transitions of the leapfrog product automaton $\Delta$ contains:
\begin{itemize}
\item $(p_0, q_0, 0, \sigma_0) \xrightarrow[]{\sigma_1} (p_2, q_0, 1, \sigma_1)$ whenever $\A_0$ has the path $p_0 \xrightarrow[]{\sigma_0} p_1 \xrightarrow[]{\sigma_1} p_2$ and for all $q_0 \in Q_0$,
\item $(p_0, q_0, 1, \sigma_0) \xrightarrow[]{\sigma_1} (p_0, q_2, 0, \sigma_1)$ whenever $\A_1$ has the path $q_0 \xrightarrow[]{\sigma_0} q_1 \xrightarrow[]{\sigma_1} q_2$ and for all $p_0 \in Q_1$.
\end{itemize}
Here the entries `$0$' and `$1$' indicate which component automaton
is behind in the simulation.
A sketch of $\frogprod$ is shown in \Cref{fig:leapfrog}.

This is the main idea behind the leapfrog product, which saves a factor of $\ell$ in the number of states compared to the catch-up product.
Given $k$ NFA, the leapfrog product automaton processes $k|w|$ updates while only storing $\Sigma^{k-1}$ in memory; see Figure~\ref{fig:leapfrog}.

 \begin{figure*}[t!]
    \centering

   \begin{subfigure}[t]{0.61\textwidth}
        \centering
        \includegraphics[width=0.99\textwidth]{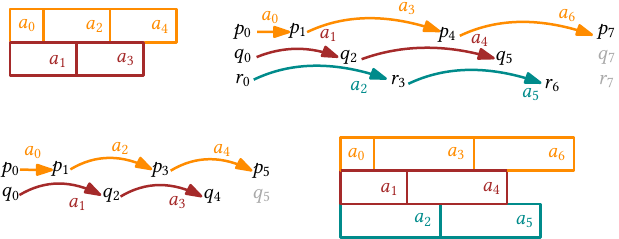} 
    \end{subfigure}
\caption{A run of $\frogprod$}    \label{fig:leapfrogrun}    

\end{figure*}

Let $m_{\scriptscriptstyle{\leq k}}$ be defined as in the previous subsection.

\begin{theorem}\label{thm:leapfrog}
    The NFA $\frogprod$: 
    \begin{itemize}
        \item recognises $\shortint$,
        \item has at most $2k\ell^{k-1} \cdot n^k $ states and at most $2k\ell^{k} \cdot m_{\scriptscriptstyle{\leq k}} \hspace{0.25em}n^{k-1}$  transitions, and 
        \item can be constructed in $\BigO(k\ell^{k} \cdot m_{\scriptscriptstyle{\leq k}} \hspace{0.25em}n^{k-1})$  time.
    \end{itemize} 
\end{theorem}

A formal definition of the NFA \frogprod, as well as a proof of \Cref{thm:leapfrog},
is given in \Cref{app:leapfrog}.

%% file: deciding.tex
\section{Deciding emptiness of \mbox{NFA intersection}: \\Algorithm and complexity}

\subsection{Upper bound}\label{sec:UB}

Using the nodding product automaton, we give a new algorithm for solving NFA $k$-intersection emptiness that is faster than constructing the direct product when the automata are dense.

\begin{theorem}\label{thm:UB}
    There is an algorithm for NFA $k$-IE that runs in {$\BigO(k \cdot m n^{k-1})$} time. 
\end{theorem}
\begin{proof}
	We construct the accessible part of the nodding product automaton $\nodprod$. 
	By \Cref{thm:nodding}, NFA $k$-IE reduces to deciding if $\Lang(\nodprod) = \emptyset$. 

    For this, our algorithm checks if its final states are reachable from its initial state, using breadth-first search. 
    The running time of this algorithm is bounded by the running time of constructing the accessible part of $\nodprod$. By \Cref{thm:nodding}, this is $\BigO(k \cdot n^{k-1}m)$ time. 
\end{proof}
In particular, consider two NFA over a fixed alphabet. Deciding reachability in $\nodprod$ takes $\BigO(mn) = \BigO(n^3)$ time in the worst case, so this is a cubic time algorithm for 2-IE. 

	\Cref{thm:UB} gives an algorithm that decides, given $k$ NFA each having $n$ states, whether they accept a word in common in 
   \begin{itemize}
       \item $\BigO(kn^{k+2})$ time, when  $|\Sigma| = \BigO(n)$, and 
        \item $\BigO(kn^{k+1})$ time, when  $|\Sigma| = \BigO(1)$. 

   \end{itemize}

\subsection{Lower bound}\label{sec:LB}
\begin{figure*}[t] 
    \hspace{1em}\begin{subfigure}[t]{0.24\textwidth}
        \centering
        \includegraphics[width=0.99\textwidth]{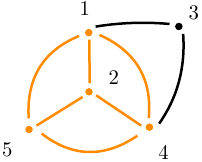}
    \caption{A graph $G$ that contains a $4$-clique}
    \label{fig:graph}
    \end{subfigure}
~ \hspace{-2ex}
    \begin{subfigure}[t]{0.18\textwidth}
    \centering
\includegraphics[width=0.99\textwidth]{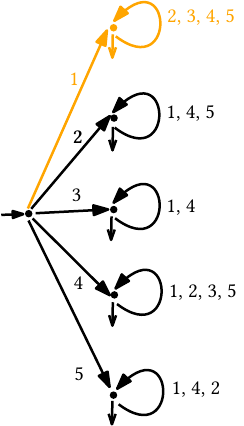}
    \caption{DFA $\A_0$}
    \label{fig:a0}
    \end{subfigure}
~ \hspace{-5ex}
\begin{subfigure}[t]{0.242\textwidth}
    \centering
\includegraphics[width=0.99\textwidth]{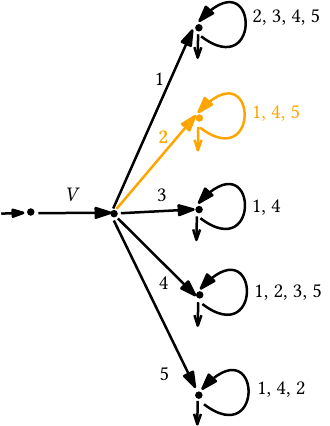}
    \caption{DFA $\A_1$}
    \label{fig:a1}
\end{subfigure}
~ \hspace{-5ex}
\begin{subfigure}[t]{0.32\textwidth}
    \centering
\includegraphics[width=0.96\textwidth]{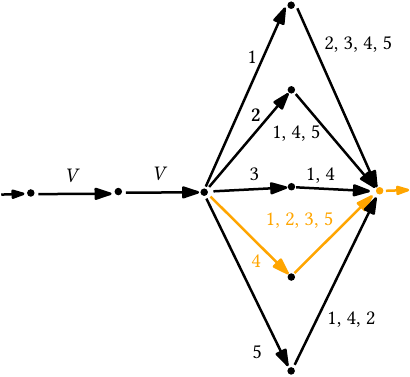}
    \caption{DFA $\A_2$}
    \label{fig:a2}
\end{subfigure}
\caption{Graph $G=(V,E)$ and DFA $\A_0, \A_1, \A_2$ such that $G$ contains a $4$-clique
iff these $3$ DFA accept a word in common. Every edge labelled $V$ represents five transitions.}
\end{figure*}
In this section, we give a lower bound on the complexity of NFA $k$-IE.  
We show that every algorithm for NFA $k$-IE either is not much faster than $\BigO( m n^{k-1})$,
or uses fast matrix multiplication.
At the source of our reduction is the following well-established hypothesis in fine-grained complexity:

 \begin{hypothesis}[Combinatorial $k$-Clique Hypothesis~\cite{AbboudBW15a}]\label{hyp:kclique}
    There is no combinatorial $\BigO(n^{k-\epsilon})$-time algorithm that can detect whether a given undirected graph on $n$ vertices contains a clique of size $k$, that is, a set of $k$ pairwise adjacent vertices. 
     \end{hypothesis}
     
     The notion of a `combinatorial' algorithm is not rigorous. It refers to avoiding fast matrix multiplication. Many such algorithms have only a small constant in the Big~O bound on the worst-case running time and use natural combinatorial structures based on the input~\cite{AbboudFischerSchechter}. 
Subject to \Cref{hyp:kclique},  multiple problems have been proved resistant to `combinatorial' speed-ups.  Without the `combinatorial' qualifier, \Cref{hyp:kclique} would be false, since $k$-cliques can be detected using an algorithm with running time bounded above by $\BigO(n^{\omega \lceil k/3 \rceil})$ which uses matrix multiplication. For example, the algorithm from \Cref{thm:UB} is combinatorial. 

\begin{lemma}\label{lem:kclique} Given a number $k > 2$ and a graph $G = (V, E)$ where $|V| = n$ and $|E| = m$, there exists an algorithm that constructs $(k-1)$ DFA such that:
\begin{enumerate}
\renewcommand{\theenumi}{(\roman{enumi})}
\renewcommand{\labelenumi}{\theenumi}
\item they all accept a word in common if and only if $G$ has a $k$-clique,
\item they each have $\BigO(n)$ states, $\BigO(m)$ transitions, and a $\BigO(n)$ sized alphabet, and
\item the running time of this algorithm is $\BigO(k \cdot n^2)$. 
\end{enumerate}
\end{lemma}

\begin{example}\label{ex1}
    Given the graph $G$ in~\Cref{fig:graph}, we  construct $3$ DFA whose intersection contains a word $w$ if and only if $w$ encodes a $4$-clique such as $\{1, 2, 4, 5\}$. 
Let the vertices be $V = \{1, 2, 3, 4, 5\}$. Our encoding maps sequences of vertices (like $(1, 2, 4, 5)$) to corresponding strings over the alphabet $\{1, 2, 3, 4, 5\}$ (like $1245$). 
\begin{enumerate}
    \item The first DFA, $\A_0$, (see \Cref{fig:a0}), ensures that the first vertex read is adjacent to all of the vertices that follow. For example, $\A_0$ accepts $1245$ since $1$ is adjacent to $2, 4, \text{ and } 5$. 
\item The second DFA, $\A_1$, (see \Cref{fig:a1}) ensures that the second vertex read is adjacent to all of the vertices that follow. For example, $\A_1$ accepts $1245$ since $2$ is adjacent to $4$ and $5$.  
\item  The third DFA, $\A_2$, (see \Cref{fig:a2}) ensures that the third vertex read is adjacent to the fourth vertex, and that there are only four vertices. For example, $\A_2$ accepts $1245$ because $4$ is adjacent to $5$. 
\end{enumerate}
Overall,  
$\A_0, \A_1\text{ and }  \A_2$ accept a word in common if and only if the word encodes a $4$-clique in $G$. 
\end{example}

\begin{proof}[Proof of \Cref{lem:kclique} (Sketch)]

Let $\Sigma = V$ be an alphabet. For every $k\geq 3$, our encoding maps sequences of $k$ vertices $v_0, v_1, \dots,  v_{k-1}$ to words of the form $v_0v_1\dots v_{k-1}$.

A sequence of vertices forms a $k$-clique whenever it satisfies the following $k-1$ conditions, where the $i$th condition is that  $$v_{i} \text{ is adjacent to each } v_j \text{ for }j \in [i+1,k-1].$$
For each of these conditions, we construct a DFA which verifies whether the word $w = v_0 v_1\dots v_{k-1}$ satisfies the condition. 
Note that the last DFA must ensure that exactly one vertex is read after $v_{k-2}$. 

For every $i \in \{0, 1, 2, \dots, k-2\}$, $\A_i$ on reading $w$ checks that $v_i$ is adjacent to all of the vertices after it (see \Cref{fig:a1}). 
 Each $\A_i$ consists of a chain of $i+1$ states, starting at the initial state, where each consecutive pair is connected with a $v$ edge for all vertices $v \in V$. At the $(i+1)$th state, on reading a letter in  $V$, the automaton goes to one of $n$ possible new states, one for each vertex. In particular, on reading some $v$ it goes to a state $q_v$, which has self-loops on letters $v'$ whenever $v$ is adjacent to $v'$, and no other outgoing edges. For every vertex $v \in V$, the state $q_v$ is final. 
The first $(k-2)$ DFA each have at most $n+k$ states, and $kn + n + 2m$ transitions.

The last DFA, $\A_{k-2}$, (see \Cref{fig:a2}) consists of a chain of $k-1$ states, starting at the initial state, with each consecutive pair in the chain being connected with a $v$ edge for all vertices $v \in V$. 
At the $(k-1)$th state, on reading a letter in  $V$, the automaton goes to one of $n$ possible new states, one for each vertex. In particular, on reading some $v$ it goes to the state $q_v$. 
At $q_v$, there is an edge to the unique final state $q_f$ labelled by the letter $v'$ whenever $v$ is adjacent to $v'$. There are no other outgoing edges from $q_v$. Hence, $\A_{k-2}$ ensures that overall exactly $k$ vertices are read, and that the last two are adjacent. This DFA has $k-2 + 2n$ states and $(k-2)n + 2m$ transitions.

Together, these $(k-1)$ automata ensure that a word is accepted in common if and only if the input graph contains a $k$-clique. 
All of these DFA have $\BigO(n)$ states, $\BigO(m)$ transitions and a $\BigO(n)$ sized alphabet.
Moreover, all of these automata can be constructed in $\BigO(km)$ time. 
\end{proof}

\begin{theorem}\label{thm:LB}
Let $k > 2$ be fixed. If for some real $\varepsilon > 0$ there exists a combinatorial $\BigO((k \cdot mn^{k-1})^{1-\varepsilon})$ time algorithm for DFA $k$-IE, then the~{Combinatorial $k$-Clique Hypothesis} is false. 
\end{theorem}
\begin{proof}
	Suppose there is an $\BigO((mn^{k-1}) ^{1-\varepsilon})$ time combinatorial algorithm for solving DFA $k$-intersection emptiness. 
		By \Cref{lem:kclique} we can reduce the $(k+1)$-Clique problem in $\BigO(km)$ time to the intersection emptiness problem for $k$ DFA each having $\BigO(n)$ states, $\BigO(m)$ transitions and a $\BigO(n)$ sized alphabet. 
	We can solve this in $\BigO((n^2\cdot n^{k-1})^{(1 -\varepsilon)}) = \BigO(n^{(k+1)(1 -\varepsilon)})$ time, by our assumption, which gives a $\BigO(n^{(k+1)(1 -\varepsilon)})$ time algorithm for detecting $(k+1)$-cliques, also combinatorial. 
\end{proof}

Naturally, a faster combinatorial algorithm for NFA $k$-IE would contradict the~{Combinatorial $k$-Clique Hypothesis} as well.  

\begin{corollary}\label{cor:LB}
Let $k > 2$ be fixed. If for some real $\varepsilon > 0$ there exists a combinatorial $\BigO((k \cdot mn^{k-1})^{1-\varepsilon})$ time algorithm for NFA $k$-IE, then the~{Combinatorial $k$-Clique Hypothesis} is false. 
\end{corollary}

%% file: certificates.tex
\section{Certifying (non-)emptiness of \mbox{NFA intersection}}
\label{sec:certificates}

In this section, we show that the nodding product construction does not just give a faster algorithm for NFA $k$-IE, but also leads to a faster algorithm for verifying whether $k$ NFA have an empty intersection or not, using an appropriate \emph{certificate}. A certificate is an object that acts as a proof of the correctness of the algorithm's output. It can obtained as a  byproduct of the NFA $k$-IE decision procedure. We describe what these certificates look like, and algorithms for verifying the output of an instance of NFA $k$-IE given a certificate. For more on certification and certifying algorithms, see~\cite{DBLP:journals/csr/McConnellMNS11}.

We  show that verifying suitable certificates for NFA $k$-IE is faster than the decision procedure for NFA $k$-IE. More specifically, NFA $k$-IE has certificates that can be verified in $\BigO(k \ell \cdot  n^{k + \omega -2})$ time, which is strictly faster than our deterministic $\BigO(k \cdot m n^{k-1})$-time algorithm whenever $\ell n^{\omega-1} = o(m)$. 
For example, when $k=2$, $m = \BigO(n^2)$ and $\ell = \BigO(1)$, i.e., for two dense NFA over a constant alphabet, given a purported certificate, the verification algorithm runs in time $\BigO(k \ell n^{2 + \omega - 2}) = \BigO(n^{\omega})$ time while the decision procedure takes $\BigO(mn) = \BigO(n^3)$ time. 

The first subsection below is a warm-up towards our main contribution in this section
(which is described in \Cref{sec:staggered}).
Throughout the section, our NFA are of the form $\A_i = (Q_i, \Sigma, \Delta_i, s_i, F_i)$.

\subsection{Certifying NFA intersection non-emptiness: Short pathsets}

Given $k$ NFA $\knfa$, the certificate for intersection \emph{non-emptiness} is a sequence of $k$ accepting runs on a short enough word, which we call a short pathset. 

\begin{definition}Given $k$ NFA $\knfa$ with at most $n$ states each, a \emph{short pathset} is a sequence $(\rho_0, \dots, \rho_{k-1})$ of $k$ strings of transitions such that each $\rho_i$ is an accepting run in $\A_i$, they are all labelled by the same word $w$, and $|w| \leq n^k$.  
\end{definition}

\begin{observation}\label{thm:Negative certificates}
    Given NFA $\knfa$ with at most $n$ states in each:
    \begin{enumerate}
        \item If $\shortint \neq \emptyset$, there exists a short pathset. 
        \item If $\shortint= \emptyset$, no pathset exists.
        \item There exists an algorithm that, given $\knfa$ and $k$ strings of transitions $\rho_0, \dots , \rho_{k-1}$, verifies in  $\BigO(k \ell \cdot n^k)$ time  whether  $(\rho_0, \dots , \rho_{k-1})$ is a short pathset for $\knfa$. 
    \end{enumerate} 
\end{observation}

    For the running time bound, note that for each $i \in [k]$, if $\rho_i = (t_0, \dots , t_{|w|-1})$, we need to verify that each $t_j$ is a transition of $\Delta_i$. To do so, we construct the adjacency matrices for each NFA and input letter and look up the entry in the adjacency matrix $\Delta_i(\sigma)$, $\sigma \in \Sigma$, where $\sigma$ is the label of $t_j$. This takes $\BigO(\ell)$ time for each transition, and we can verify in advance that the runs are all at most $n^k$ transitions long.

Thus, the intersection non-emptiness of $k$ NFA can be \emph{certified} in $\BigO(k \ell \cdot n^k)$ time. An alternative view is that this is a co-nondeterministic algorithm for NFA $k$-IE. 

\begin{corollary}\label{cor:contime}
NFA $k$-IE is in $\mathsf{coNTIME}(k\ell \cdot n^k)$. 
\end{corollary}

\begin{proof}
We give a nondeterministic algorithm for \emph{non-emptiness}.
Given NFA $\knfa$, if $\shortint $ $ \ne \emptyset$, the algorithm guesses
and checks a short pathset for them.
\Cref{thm:Negative certificates} provides the correctness guarantees
and the running time bound.
\end{proof}

\subsection{Certifying NFA intersection emptiness: Staggered cuts}
\label{sec:staggered}
We will now describe how to efficiently certify that the intersection \shortint for \knfa is \emph{empty}.
Our goal is to conveniently represent a cut in the transition graph of an automaton for \shortint.  Such a cut exists if and only if the intersection is empty. Here, we use the nodding product automaton $\nodprod$ since it is the smallest, and its transition graph allows us to efficiently represent and verify such a cut. 

Naively, a cut is a subset $R$ of the set of states such that:
\begin{enumerate}
\renewcommand{\labelenumi}{\theenumi)}
    \item $R$ contains the initial state,
    \item $R \cap F = \emptyset$, where $F$ is the set of final states, and
    \item no transition in $\Delta$ leads from $R$ to its complement.
\end{enumerate}
Given such a subset, the first two conditions can be checked with one pass, and the third can be checked by scanning through the transition relation $\Delta$ of \nodprod. Unfortunately, since  $\Delta_{\text{nodding}}$ has size $\Theta(mn^{k-1})$, this approach takes at least as long as deciding NFA $k$-IE. 
Instead, we utilise the structure of $\nodprod$. We check that no transition leaves $R$ using fast matrix multiplication with the $n \times n$ sized adjacency matrices of $\Delta_i$ instead of accessing $\Delta$ directly.

\begin{definition}
Given $k$ NFA $\knfa$ over an alphabet $\Sigma$, a \emph{staggered cut} is a collection $$(R(i, \sigma) :\ \text{for all $i \in [k]$ and all $\sigma \in \Sigma$})$$ where for each $i \in [k]$ and for all letters $\sigma \in \Sigma$:
	\begin{enumerate} 
        \item $R(i, \sigma) \subseteq Q_0 \times \dots  \times Q_{k-1}$.
        \item For all $\sigma, \sigma' \in \Sigma$, $R(0, \sigma) = R(0, \sigma')$.
        \item For all $\sigma \in \Sigma$, $(s_0, \dots , s_{k-1}) \in R(0, \sigma)$. 
		\item For all $\sigma \in \Sigma$,  $ (F_0 \times \dots \times F_{k-1}) \cap R(0,\sigma) = \emptyset$. 
		\item No transition leaves the staggered cut: for all $ \sigma \in \Sigma$ and for all $i \in [k]$, 
         \begin{multline*}
         \qquad\text{for all } \mathbf{q} = (q_0, \dots, q_{k-1}) \in R(i, \sigma), \\ \text{ if } (q_i, \sigma, p_i) \in \Delta_i,
         \text{ then }\mathbf{q'} = (q'_0, \dots, q'_{k-1})\in R(i+1 \bmod k, \sigma), \\
         \text{where $q'_j =  p_j$ if $ j = i$,  and $q'_j =  q_j $ otherwise. }
         \end{multline*}
	\end{enumerate}
\end{definition}

Note that in a staggered cut  there are $k\ell$ sets, each of size at most $n^k$
if every NFA among \knfa has at most $n$ states.

\begin{theorem}\label{lem:Positive certificates}
   Given NFA $\knfa$ with at most $n$ states in each:
    \begin{enumerate}
        \item If $\shortint = \emptyset$, there exists a staggered cut.
        \item If $\shortint \neq \emptyset$, no staggered cut exists.
\item There exists an algorithm that, given $\knfa$ and a collection of $k\ell$ subsets of $Q_0 \times \dots \times Q_{k-1}$, verifies in time $\BigO(k \ell \cdot n^{k + \omega - 2})$ whether  $R$ is a staggered cut of $\knfa$. 
    \end{enumerate} 
\end{theorem}

\begin{proof}[Proof (Sketch)]
    \begin{enumerate}
        \item If $\shortint = \emptyset$, then we know no final state of \nodprod is reachable from the initial state. Let $R(0)$ be the set of tuples $(\qvec)$ such that the base copy state $(\qvec, \varepsilon)$ is reachable from the initial state in $\nodprod$. Clearly, $(s_0, \dots, s_{k-1}) \in R(0)$. For all letters $\sigma \in \Sigma$, let $R(0, \sigma) = R(0)$. For all $\sigma \in \Sigma$ and $i \in [1, k-1]$, let $R(i,\sigma)$ be the set of tuples $(\qvec)$ such that $(\qvec, i, \sigma)$ is reachable from the initial state in $\nodprod$. Whenever $(\qvec) \in R(i, \sigma)$, we say that this tuple \emph{represents} the state $(\qvec, i, \sigma)$ for an appropriate choice of $i \in [k]$ and $\sigma \in \Sigma \cup \{\epsilon\}$. Thus, we have formed a collection of $k\ell$ subsets $R(i, \sigma)$ of $Q_0 \times \dots \times Q_{k-1}$ that contain the representative of the initial state of \nodprod, contain no representative of its final states, and since they each represent the set of all reachable states in a given copy, all transitions in $\nodprod$ whose source represents an element of $R(i,\sigma)$ must end in the representative of some its successor. Thus, our choice of $(R(i,\sigma))$ is in agreement with the definition of staggered cut.
        
        \item Suppose $\shortint \neq \emptyset$, and a staggered cut $R$ exists. Let the shortest word in $\shortint$ be $w$. By applying the transitions that form an accepting run of $w$ in $\nodprod$, one starts at the initial state and eventually reaches a final state. Since $(s_0, \dots, s_{k-1}) \in R(0, \sigma)$ for some $\sigma$, and no transition may leave $R$, by retracing the same transitions on the elements of sets of $R$, either a transition must leave $R$ or it must be the case that $R$ contains a representative of a final state. Both scenarios contradict $R$ being a staggered cut. 
        
        \item The certificate checker can check whether all the $R(0, \sigma)$ are the set $R(0)$, and that initial states are present and final states are absent in $R(0)$ in $\BigO(k \cdot n^k)$-time. In order to ensure that no transition leaves $R$, the certificate checker stores each subset in two data structures. Given any $i \in [k]$ and any $\sigma \in \Sigma$, the certificate checker stores each subset $R(i,\sigma)$ using two $n^{k-1} \times n$ dimensional adjacency matrices, $\text{In}_{i, \sigma}$ and $\text{Out}_{i, \sigma}$. Let elements of $\{0, \dots, n^{k-1}-1\}$ be represented as $(k-1)$-tuples $(j_0, \dots , j_{k-1})$ where each $j_i \leq n$.
\begin{gather*}
\text{In}_{p, \sigma}[ (j_0, \dots , j_{k-1}), i] = 
            1 \text{ iff } (q_{j_0}, \dots , q_{j_{p-2}}, q_i, q_{j_{p-1}}, \dots , q_{j_{k-1}}) \in R(p, \sigma).\\
\text{Out}_{p, \sigma}[ (j_0, \dots , j_{k-1}), i] = 
            1 \text{ iff } (q_{j_0}, \dots , q_{j_{p-1}}, q_i, q_{j_{p}}, \dots , q_{j_{k-1}}) \in R(p, \sigma).
\end{gather*}
    \end{enumerate}

The $\text{In}_{p, \sigma}$ matrix allow for easy access to the $Q_{p-1}$ component of a copy, and the $\text{Out}_{p, \sigma}$ matrix provides easy access to the $Q_p$ component of the copy. Since outgoing transitions only update the $Q_p$ component and ingoing transitions only update the $Q_{p-1}$ component, this is all we need. 

Given the $\text{In}$ and $\text{Out}$ matrices, for each $p \in [k]$, given $\Delta_p(\sigma)$ also in adjacency matrix form, i.e., $\Delta_p[i,j] = 1 \iff (q_i, \sigma, q_j) \in \Delta_p$, to ensure no transition leaves $R$ the checker ensures that $$\text{Out}_{p, \sigma} \cdot \Delta_p(\sigma) \leq \text{In}_{p+1, \sigma}$$ 

These matrices can be constructed from $R$ in $\BigO(k\ell \cdot n^k)$ time. All of these inequalities can be checked using fast matrix multiplication by splitting up the $\text{Out}$ matrices into $n^{k-2}$ many $n \times n$ matrices, so overall this check takes $\BigO(k\ell \cdot n^{k + \omega - 2})$ time.  
 \end{proof}

Hence, NFA $k$-IE has a certificate that can be verified in $\BigO(k \ell \cdot n^{k + \omega -2})$ time. An alternative viewpoint is that this is a nondeterministic algorithm for  NFA $k$-IE.

\begin{corollary}\label{cor:ntime}
NFA $k$-IE is in $\mathsf{NTIME}(k \ell \cdot n^{k + \omega - 2})$. 
\end{corollary}

\begin{proof}
Given $k$ NFA $\knfa$ such that $\shortint = \emptyset$, the nondeterministic algorithm guesses
and checks a staggered cut.
\Cref{lem:Positive certificates} provides the correctness guarantees
and the running time bound.
\end{proof}
 
\begin{remark*}
When the languages of $\knfa$ have non-empty intersection (empty intersection, respectively), our decision procedure for NFA $k$-IE (algorithm of \Cref{thm:UB}) can be modified to store and output a short pathset (a staggered cut, respectively) for the instance without adding to the running time. 
\end{remark*}

\subsection{A tight SETH lower bound for NFA $k$-IE is unlikely}

We already mentioned in \Cref{sec:intro}
that an $O(n^{k-\epsilon})$-time algorithm for DFA (and thus NFA) $k$-IE
would refute 
SETH, the Strong Exponential Time Hypothesis, by a result of
Wehar~\cite[Cor.~7.19]{wehar2017complexity}.
Our lower bound of $m n^{k-1}$ (\Cref{thm:LB}) is based
on the Combinatorial $k$-Clique Hypothesis~\cite{AbboudBW15a}.
SETH is, in comparison,
more widely used.
One can ask whether it is also possible to establish hardness results
similar to \Cref{thm:LB} but based on SETH.
We show that, in fact, a tight SETH-based lower bound for NFA $k$-IE is unlikely. This is a consequence of our previous result, that our certificate checkers for NFA $k$-IE (\Cref{cor:contime} and \Cref{cor:ntime}) run faster than the deterministic algorithm for NFA $k$-IE (\Cref{thm:UB}).

We need a notion of reduction that distinguishes between
the running times $\BigO(n^k)$ and $\BigO(n^{k+1})$.

Let pairs $(\Lang,T)$ correspond to assertions
that (membership in) a language $\Lang$ can be decided in $T$~time.

\begin{definition}[Fine-grained reduction~\cite{WilliamsFG2018}]
	The pair $(A, a(n))$ is said to be fine-grained reducible to $(B, b(n))$, which we denote as \begin{equation*}(A, a(n)) \leq_{FGR} (B, b(n)),\end{equation*} if for every $\epsilon > 0$ there exists $\delta > 0$ such that there is a deterministic Turing reduction from $A$ to $B$ which runs in time at most $a(n)^{1-\delta}$, making $c$ calls to an oracle for $B$ with query lengths $n_0, \dots , n_{c-1}$, where
$\sum_{i=0}^{c-1} b(n_i)^{1-\epsilon} \leq a(n)^{1-\delta}$. 
\end{definition}

The existence of such a reduction means that for any real $\epsilon > 0$, if there was a $b(n)^{1-\epsilon}$ time algorithm for deciding $B$, then there exists $\delta > 0$ such that $A$ can be decided in $a(n)^{1-\delta}$ time. In this way, we can transfer savings in the exponent of one problem's running time, to the other problem.

\begin{hypothesis}[SETH = Strong Exponential Time Hypothesis~\cite{IP01}; \textup{see also \cite{WilliamsFG2018,Bringmann19}}]
        For every $\epsilon > 0$, 
        there exists $k$ such that $k$-SAT is not in $\mathsf{DTIME}[2^{n(1-\epsilon)}]$, where $k$-SAT is the language of all satisfiable Boolean formulas in $k$-CNF.
         \end{hypothesis}

The following hypothesis strengthens SETH even further.

		 \begin{hypothesis}[NSETH = Nondeterministic Strong Exponential Time Hypothesis~\cite{CarmosinoGIMPS16}]
       For every $\epsilon > 0$, 
        there exists $k$ such that $k$-TAUT is not in $\mathsf{NTIME}$$[2^{n(1-\epsilon)}]$, where $k$-TAUT is the language of all tautological Boolean formulas in $k$-DNF.
		     
		 \end{hypothesis}

A refutation of NSETH would mean a `more efficient' proof system for
propositional tautologies, a breakthrough in proof complexity.
This would also imply new circuit lower bounds~\cite{CarmosinoGIMPS16}.

We will now state Corollary 2 of Theorem 2 in~\cite{CarmosinoGIMPS16}, where the authors show that if a problem has a certification scheme that runs in, say, time $\BigO(n^i)$, then an SETH-based lower bound stronger than $\Omega(n^i)$ would contradict NSETH. 

\begin{theorem}[Carmosino et al., 2016]\label{thm:SETHnotNSETH}
Given a problem $C \in \mathsf{NTIME}(T_C) \cap \mathsf{coNTIME}(T_C)$, if NSETH holds, then $(B, T_B) \not \leq_{FGR} (C, T_C^{1+\gamma})$ for every decision problem $(B, T_B)$ such that $(\textsc{SAT}, 2^n) \leq_{FGR} (B, T_B)$ and every $\gamma > 0$. 
\end{theorem}

Because we have a certificate scheme whose checker is faster than the deterministic algorithm for NFA $k$-IE, then, assuming NSETH,  $(\textsc{SAT}, 2^n) \not \leq_{FGR} (\text{NFA }k\text{-IE}, \BigO(n^{k-1}m)$). The idea behind \Cref{thm:SETHnotNSETH} is that a fine-grained reduction from SAT could be composed with the certificate checkers, to yield a nondeterministic algorithm for TAUT.

For our problem, such a reduction ensures that any faster algorithm for NFA $k$-IE yields a strictly faster than $\mathsf{DTIME}[2^n]$ algorithm for SAT, as well as for TAUT. When composed with the certificate checkers, each oracle call to the certificate checkers runs strictly faster than the oracle calls to NFA $k$-IE's decision procedure. Hence, this would result in the overall running time of the nondeterministic algorithm being strictly less than $2^{n}$, and in fact $O(2^{n (1 - \epsilon)})$.

\begin{theorem}\label{cor:NFAnotSETH}
If NSETH holds, then for every $\gamma > 0$ and for all decision problems  $(B, T_B)$ such that $(\textsc{SAT}, 2^n) \leq_{FGR} (B, T_B)$, we have $(B, T_B) \not \leq_{FGR} ( \text{NFA }k\text{-IE}, (\ell n^{k + \omega -2})^{1+\gamma})$.
\end{theorem}
\begin{proof}
By \Cref{cor:contime} and \Cref{cor:ntime}, NFA $k$-IE belongs to $\mathsf{coNTIME}(k \ell \cdot n^{k + \omega - 2})\, \cap \, \mathsf{NTIME}(k \ell \cdot n^{k + \omega - 2})$. 

By \Cref{thm:SETHnotNSETH}, no SETH-based lower bound can be stronger than $\BigO(k \ell \cdot n^{k + \omega - 2})$, unless NSETH fails. 
Whenever $\ell n^{\omega -1} = o(m)$, this means that a tight SETH lower bound for our $\BigO(n^{k-1}m)$ time algorithm contradicts NSETH. 
\end{proof}

In words, a stronger than $\ell n^{k+\omega-2}$ lower bound for NFA $k$-IE based on SETH
is unlikely: it would refute NSETH.

%% file: extras.tex
\section{Discussion}

\paragraph{Transition complexity.}
We have constructed NFA (and $\epsilon$-NFA) for
languages of the form \shortint, where \knfa are NFA
with at most $n$ states in each.
The advantage over the standard direct product construction
is that our product automata are sparser, that is,
have fewer transitions.
An interesting open question is whether even sparser
constructions exist. 
Proving lower bounds on the number of transitions
in nondeterministic automata is a delicate task;
we refer the reader to publications~\cite{GruberH07,HromkovicS07}
and to Section~2.1.2 in the book chapter~\cite{GruberHK-handbook}
for an introduction to the state of the art.

Let an NFA recognising a language $L$ be called transition-minimal if, among all the NFA recognising $L$, it has the fewest number of transitions. The transition complexity of a language $L$ is the number of transitions in a transition-minimal NFA that recognises $L$.  
Given two regular languages, the transition complexity of their intersection does not exceed the product of their individual transition complexities~\cite{DomaratzkiS07}. 
Our \Cref{thm:nodding} shows that this bound is not optimal when the two languages have dense transition-minimal NFA. 

Comparing \Cref{thm:nodding} with the standard direct product also suggests the following interpretation: to obtain an automaton for the intersection of two languages, if one wants to reduce the number of transitions, one pays a factor of $|\Sigma|$ in the number of states. 
This is reminiscent of a conjecture due to Domaratzki and Salomaa~\cite{DomaratzkiS08}: Given a regular language $L$ over an alphabet $\Sigma$, let the minimum number of states an NFA needs to recognise the language be $n$. Then, the NFA recognising $L$ that has the least number of transitions, must have at most $|\Sigma|\cdot n$ states. 

The automata delivered by \Cref{thm:nodding} have $\epsilon$-transitions, but
our other two product constructions, the catch-up and leapfrog product, due to being $\epsilon$-free NFA, are much larger than the nodding product when the alphabet of interest grows. Still, the catch-up and leapfrog products can be smaller than the direct product in the worst case, even when the alphabet grows. For an alphabet of size $\ell$, over $n$ states, the (dense) worst case for the upper bounds of Theorems~\ref{thm:catch-up} and~\ref{thm:leapfrog} is when $m = \ell n^2$ and $m_{\leq k} = n^2$. Let $k$ be fixed. The size (number of states and transitions) of the direct product is then $\BigO(\ell n^{2k})$, while the size of both the catch-up and the leapfrog automata is bounded above by $\BigO(\ell^{k}n^{k-1}m_{\leq k}) = \BigO(\ell^kn^{k+1})$. If $\ell = o(n)$, then $\ell^k n^{k+1} = o(\ell n^{2k})$. Hence, whenever the number of letters in $\Sigma$ grows slower than the bound $n$ on the number of states, the worst case bounds for all of our products improve upon those of the direct product.

\paragraph{Fine-grained complexity.}
Fix $k \ge 2$.
For the intersection emptiness problem for $k$ DFA,
Wehar~\cite{wehar2017complexity} proves that, for every real $\epsilon > 0$,
an $n^{k-\epsilon}$-time deterministic algorithm would
falsify SETH.
At the same time, best existing algorithms for this problem run
in time $O(\ell \cdot n^k)$, where $\ell = |\Sigma|$;
which is the same as $O(m n^{k-1})$ if the DFA are complete.
If $\ell$ may grow, there is a gap between the lower and upper bound.
One of our contributions is to close this gap, strengthening the lower bound.
Theorem~\ref{thm:LB} shows that an $(m n^{k-1})^{1-\epsilon}$-time algorithm
(already for the problem on DFA, not NFA)
would falsify the Combinatorial $k$-Clique Hypothesis.
Our hard instances have automata with $|\Sigma| = \Theta(n)$.
Thus, it is the \emph{density} of the state-transition diagrams
(in particular, $m$ may be as high as $n^2$) that our reduction exploits.
Previously known reductions do not give tight bounds
for the case $|Q| = o(|\Delta|)$, that is, with
the dense state-transition diagram.

In the case $|\Sigma| = O(1)$, existing lower bounds
for the intersection emptiness problem for $k$ DFA are tight.
For NFA, there was previously a gap between $\Omega(n^k)$ and $O(m^k)$.
Our algorithmic result improves the upper bound to $O(m n^{k-1})$ but
the bounds are still a factor of $m/n$ apart.
If $k = 2$, the recent NFA acceptance hypothesis~\cite{BringmannGKL24} is that
the running time of $O(m n)$ is necessary even if one
of the two NFA accepts exactly one word.

\paragraph{Relation satisfaction problems.}
The starting point of our constructions in \Cref{sec:sparser}
is \Cref{obs:fact}. The \emph{$k$-stuttering} of the intersection of $k$ regular languages\footnote{`$2$-stuttering' means that each letter is duplicated in all words; with $k$, we replace each $a \in \Sigma$ with $a^k$; see the end of Sec.~\ref{sec:nod}.} is the same as the interleaving of the languages intersected with $(a^k + b^k)^*$. This is because the $k$-tape automaton that realises the equality relation between regular languages is very similar to an NFA that recognises $(a^k + b^k)^*$. By `$k$-tape automaton' we refer to a nondeterministic automaton with a finite set of control states and $k$ one-way read-only tapes, where each transition also specifies one of the $k$ tapes, from which an input letter is read.

\Cref{obs:fact} can be generalised to check whether $k$ regular languages `satisfy' a \emph{rational relation}~\cite{DBLP:journals/eatcs/Choffrut06}. Rational relations are exactly the relations recognised by one-way $k$-tape automata (for more details see Chapter 4, section 1.6 of~\cite{Sakarovitch_2009}).

We define the $k$-Relation Satisfaction problem ($k$-RS) as follows. Fix $k$. Given $k$ NFA $\A_1, \dots ,A_k$ as input, with $n$ states and $m$ transitions, over an alphabet of size $\ell$, and a $k$-tape automaton $\mathcal{C}$ with $n_{\mathcal{C}}$ states and $m_{\mathcal{C}}$ transitions. NFA $\A_1, \dots,A_k$ are said to \emph{satisfy} $\mathcal{C}$ if there exist words $w_1 \in \Lang(\A_1), \dots, w_k \in \Lang(\A_k)$ such that $(w_1, w_2, \dots, w_k) \in \Lang(\mathcal{C})$.

NFA $k$-IE is a special case of the more general NFA $k$-RS problem, but
in fact, when the alphabet is fixed, NFA $k$-RS is inter-reducible with NFA $(k+1)$-IE in a `fine-grained' sense; see \Cref{app:controlled}.

\paragraph{Further directions.}
Our constructions of sparser product automata
may be applied to automata other than NFA.
For verification, B\"uchi automata as well as other automata
on infinite words are a natural target, and the idea
of interleaving several runs should apply.
Automata with data structures, such as pushdown automata,
can be explored, but we have not identified fine-grained algorithmic
improvements so far.

We also leave it for future work to explore applications in graph databases.
Queries to graph databases can often be interpreted in automata-theoretic terms.
In the recent years, fine-grained lower bounds for regular path query evaluation
were obtained by Casel and Schmid~\cite{CaselS23};
and for Boolean conjunctive queries by Fan, Koutris, and Zhao~\cite{FanKZ23}.
(Casel and Schmid also discuss the `product graph approach', but when a query
arises from a regular expression, the corresponding automaton is sparse.)
We refer the reader to Durand's survey~\cite{Durand20} for a broader overview.
A popular type of query related to the \emph{intersection} of regular languages
is \emph{conjunctive} regular path queries; see, e.g., Barcelo et~al.~\cite{BarceloFL13}.

%% file: cartesian.tex
\section{The standard Cartesian product automaton}\label{sec:CPA}

\begin{definition}[Direct product automaton]
    Given $k$ NFA $\knfa$ such that for all $i \in [k-1]$ $\A_i = (Q_i, \Sigma, \Delta_i, s_i, F_i)$, the direct product automaton $\A_C = (Q, \Sigma, \Delta, \mathbf{s}, F)$ is defined as follows: 
\begin{enumerate}
    \item $Q = Q_0 \times Q_1 \times \dots \times Q_{k-1}$, where $\times$ denotes the Cartesian product of sets. 
    \item $\mathbf{s} = (s_0, s_1, \dots , s_{k-1})$. 
    \item $F = F_0 \times F_1 \times \dots \times F_{k-1}$. 
\item $\Delta = \{((\qvec), \sigma, (q'_0, \dots , q'_{k-1})) \mid \text{$(q_i, \sigma, q'_i) \in \Delta_i$ for all $i \in [k-1]$} \} $. 
\end{enumerate}
\end{definition}

For a fixed $k$, given $k$ NFA ${\knfa}$ with at most $n$ states and $m$ transitions each, over an $\ell$-letter alphabet, the direct product automaton $\A_C$ contains $\BigO(n^k)$ states and $\BigO(m^k)$ transitions. 

%% file: proof-shuffle-observation.tex
\section{Proof of Observation 1}\label{sec:interleavingproof}
In this section, we prove the observation used in \Cref{sec:sparser} as the basis of our sparser automata for recognising the intersection. 

Let $\klangs \sset \{a_1, \ldots, a_\ell\}^*$. Then\vspace*{-1ex}
\begin{equation*}
\bigcap_{i=0}^{k-1} \Lang_i = [(\Shuffle_{i=0}^{k-1} \Lang_i ) \cap (a_0^k + \dots + a_{\ell-1}^k)^*]\mathlarger{{\upharpoonright}}_{0 \bmod k}\;.
\end{equation*} 

\begin{proof}
    Let a word $w$ belong to  the intersection $\bigcap_{i=0}^{k-1} \Lang_i$, i.e., for all $i\in [k-1]$, $w \in \Lang(\A_i)$. Let $w'$ be the unique word that belongs to $ \Shuffle_{i=0}^{k-1} \{w\}$. We see that $\Shuffle_{i=0}^{k-1} \{w\} = \{w'\} \subseteq (\Shuffle_{i=0}^{k-1} \Lang_i )$. At the same time, $\Shuffle_{i=0}^{k-1} \{w\} = \{w'\}$ also belongs to $(a_0^k + \dots + a_{\ell-1}^k)^*$, since within every $k$-length block of letters, all the letters match up. Hence,  $w' \in [(\Shuffle_{i=0}^{k-1} \Lang_i ) \cap (a_0^k + \dots + a_{\ell-1}^k)^*]$, and since $\{w\} = [\Shuffle_{i=0}^{k-1} \{w\} ]\mathlarger{{\upharpoonright}}_{0 \bmod k}$, $$\bigcap_{i=0}^{k-1} \Lang_i \subseteq [(\Shuffle_{i=0}^{k-1} \Lang_i ) \cap (a_0^k + \dots + a_{\ell-1}^k)^*]\mathlarger{{\upharpoonright}}_{0 \bmod k}.$$

    Conversely, let a word $w$ belong to $[(\Shuffle_{i=0}^{k-1} \Lang_i ) \cap (a_0^k + \dots + a_{\ell-1}^k)^*]\mathlarger{{\upharpoonright}}_{0 \bmod k}$. There must exist some $w' \in [(\Shuffle_{i=0}^{k-1} \Lang_i ) \cap (a_0^k + \dots + a_{\ell-1}^k)^*]$ such that $w = w' {\upharpoonright_{0 \bmod k}}$. Since $w' \in (a_0^k + \dots + a_{\ell-1}^k)^*$, for all $i \in [k-1]$, $w'{\upharpoonright_{i \bmod k}}$ are all equal to $w$. Since $w' \in \Shuffle_{i=0}^{k-1} \Lang_i  $, this means for all $i \in [k-1]$, $w\in \Lang(\A_i)$, which proves that $w \in \bigcap_{i=0}^{k-1} \Lang_i$. Hence, $$\bigcap_{i=0}^{k-1} \Lang_i \supseteq [(\Shuffle_{i=0}^{k-1} \Lang_i ) \cap (a_0^k + \dots + a_{\ell-1}^k)^*]\mathlarger{{\upharpoonright}}_{0 \bmod k}.$$

    We conclude that \begin{equation*}
\bigcap_{i=0}^{k-1} \Lang_i = [(\Shuffle_{i=0}^{k-1} \Lang_i ) \cap (a_0^k + \dots + a_{\ell-1}^k)^*]\mathlarger{{\upharpoonright}}_{0 \bmod k}\;.
\qedhere
\end{equation*} 
\end{proof}

%% file: leapfrogproof.tex
\section{The leapfrog product construction}
\label{app:leapfrog}
\begin{figure}
    \centering
    \includegraphics[width=0.55\linewidth]{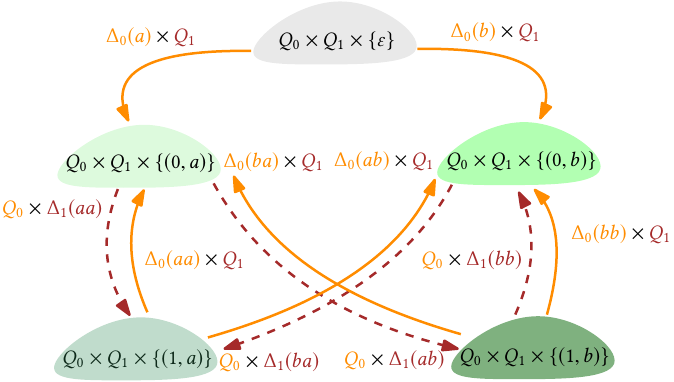}
    \caption{Leapfrog automaton with state-update labels}
    \label{fig:leapfrogappx}
\end{figure}
\subsection{Intersection of two NFA} Given $2$ NFA $\A_0$ and $\A_1$, the leapfrog product automaton $\A_{\text{leapfrog}}$ is defined as follows: $$\A_{\text{leapfrog}} = (Q, \Sigma, \Delta, s, F)$$ Where the states $Q$ consists of $2$ copies of $Q_0 \times  Q_{1}$ for each letter, as well as one copy of $Q_0 \times  Q_{1}$ for the empty word. That is,  $Q = Q_0 \times Q_1 \times \Sigma \times \{0,1\} \cup Q_0 \times Q_1 \times \{\varepsilon\}$ .  The $\{0, 1\}$ component in the letter labelled copies denote which automaton is currently behind in the computation. The initial state is $(s_0, s_1, \varepsilon)$. 

For each  $\sigma_0\sigma_1 \in \Sigma^2$, the set  $\Delta$ contains:
\begin{itemize}
\item $(p_0, q_0, 0, \sigma_0) \xrightarrow[]{\sigma_1} (p_2, q_0, 1, \sigma_1)$ whenever $\A_0$ has the path $p_0 \xrightarrow[]{\sigma_0} p_1 \xrightarrow[]{\sigma_1} p_2$  and $q_0 \in Q_1$. 
\item $(p_0, q_0, 1, \sigma_0) \xrightarrow[]{\sigma_1} (p_0, q_2, 0, \sigma_1)$ whenever $\A_1$ has the path $q_0 \xrightarrow[]{\sigma_0} q_1 \xrightarrow[]{\sigma_1} q_2$ and $p_0 \in Q_0$.
\end{itemize}

Additionally, for each  $\sigma \in \Sigma^2$, the set $\Delta$ contains:
$$(p_0, q_0, \varepsilon) \xrightarrow[]{\sigma} (p_1, q_0, 1, \sigma)\text{ whenever }p_{0} \xrightarrow[]{\sigma} p_{1} \text{ in }\A_{0} \text{ and } q_0 \in Q_1.$$

In words, for all $\sigma \in \Sigma$, the $\varepsilon$-labelled copy is connected to the $(1, \sigma)$-copy by a volley that reads $\sigma$ and updates the $Q_0$ component as if $\A_0$ had read $\sigma$. For all $2$ letter words $u = \sigma_0 \sigma_1$, for all letters $\sigma \in \Sigma$,  and for all $i$ in $\{0, 1\}$, the $(i, \sigma_0)$-labelled copy is connected to the $(i+1 \bmod k, \sigma_2)$ labelled copy by a volley that reads $\sigma_1$ and updates the $Q_i$ component as if $\A_i$ had read $\sigma_0\sigma_1$. 

In the $\varepsilon$-copy, $F_0\times F_1$ are the only final states. For each letter $\sigma \in \Sigma$ and for each $i$ in $\{0,1\}$, in the $(i, \sigma)$-copy, the final states are $\{(q_0, q_1) \mid  q_i \in F_i \text{ and } q_{1-i} \xrightarrow[]{\sigma} q'_{1-i} \text{ in }\A_{1-i} \text{ and } q'_{1-i} \in F_{1-i}\}$. That is, the $Q_i$ component must be accepting, and the other component must be able to read $\sigma$ and reach a final state.

\subsection{Intersection of $k \geq 2$ NFA} Given $k$ NFA $\A_0, \dots , \A_{k-1}$, the leapfrog product automaton $\A_{\text{leapfrog}}$ is defined as follows: $$\A_{\text{leapfrog}} = (Q, \Sigma, \Delta, s, F)$$ Where the states $Q$ consists of $k$ copies of $Q_0 \times \dots \times Q_{k-1}$ for each $k-1$ length word, as well as one copy of $\prod_{i=0}^{k-1}Q_i$ for every shorter word. That is,  $Q = \left((\prod_{i=0}Q_i) \times \Sigma^{k-1} \times \{0, \dots, k-1\}\right) \cup \left((\prod_{i=0}Q_i) \times \Sigma^{\leq k-2} \right)$ .  The initial state is the initial vector in the $\varepsilon$-copy. That is, $s = (s_0, \dots , s_{k-1}, \varepsilon)$. The $\{0, \dots , k-1\}$ component in the $\Sigma^{k-1}$-labelled copies denote which automaton is currently furthest behind in the computation.

\subsubsection*{Transitions of $\A_{\text{leapfrog}}$}
For all words $u$ of length up to $k-1$, the set  $\Delta$ contains:
\begin{description}
\item[If $\bm{|u| < k-2}$:\quad] $(q_0, \dots , q_{k-1}, u) \xrightarrow[]{\sigma} (q'_0, \dots , q'_{k-1}, u\sigma)$ whenever $q'_j = q_j$ for all $j \neq |u|$ and $q_{|u|} \xrightarrow[]{u\sigma} q'_{|u|}$ is a path in $\A_{|u|}$.
\item[If $\bm{|u| = k-2}$:\quad] $(q_0, \dots ,q_{k-1}, u) \xrightarrow[]{\sigma} (q'_0, \dots ,q'_{k-1}, k-1 , u\sigma)$ for all letters $\sigma \in \Sigma$, whenever  $q'_j = q_j$ for all $j \neq |u|$ and $q_{|u|} \xrightarrow[]{u\sigma} q'_{|u|}$ is a path in $\A_{|u|}$. 
\item[If $\bm{u = \sigma_0 \ldots \sigma_{k-1}}$:\quad] $(q_0, \dots ,q_{k-1}, i, u) \xrightarrow[]{\sigma} (q'_0, \dots , q'_{k-1}, i + 1 \bmod k, \sigma_1 \dots \sigma_{k-1}\sigma)$ for all letters $\sigma \in \Sigma$, whenever $q'_j = q_j$ for all $i$ in $\{0, \dots, k-1\}$, and for all $j \neq i$, and $q_{i} \xrightarrow[]{u\sigma} q'_{i}$ is a path in $\A_{i}$. 
\end{description}

In words, all of the $\Sigma^{\leq k-2}$-labelled copies are connected in a full $|\Sigma|$-ary tree starting from the $\varepsilon$ copy. Let $u$ be a word such that $|u| < k-2$.  For all $\sigma \in \Sigma$, the $u$-labelled copy is connected to the $u\sigma$-copy by a volley that reads $\sigma$ and updates the $Q_t$ component as if $\A_t$ had read $u\sigma$. Once we reach a copy labelled with a $k-2$ letter word, it is time to transition to the main body of the automaton. For any $|u| = k-2$ and any $\sigma \in \Sigma$, the $u$-copy is connected to the $(k-1, u\sigma)$-labelled copy with a volley that reads $\sigma$ and updates $\A_{k-2}$ as if it had read $u\sigma$. These transitions act as an initialisation phase, setting each consecutive NFA one step ahead of the last, and storing the entire word each component has read so far.   

Once we reach the $\Sigma^{k-1}$-labelled copies, initialisation is over, and now each component NFA will update as if it has read a $k$ length word and storing its last $k-1$ letters for the next component NFA. Lastly, for all $k-1$ letter words $u = \sigma_0 \dots \sigma_{k-1}$, for all letters $\sigma \in \Sigma$,  and for all $i$ in $\{0, \dots , k-1\}$, the $(i, u)$-labelled copy is connected to the $(i+1 \bmod k, v)$ labelled copy, where $v = \sigma_1\dots \sigma_{k-1}\sigma$, i.e., the word formed by erasing the first letter of $u$ and appending $\sigma$. The volley between the $(i, u)$ copy and the $(i+1 \bmod k)$ copy reads $\sigma$ and updates the $Q_i$ component as if $\A_i$ had read $u\sigma$. 

\subsubsection*{Final states of $\A_{\text{leapfrog}}$}
In the $\varepsilon$-copy, $F_0\times F_1$ are the only final states. For each word $u = \sigma_1 \dots \sigma_{i}$ with $i\leq k-2$, in the $u$-copy, the final states consist of states $(q_0, \dots, q_k, u)$ where $q_i$ is final, and for all $j < i$, the state  $q_j$ can reach a final state in the $Q_j$ component by reading $\sigma_{j+1} \dots \sigma_i$, and for all $j > i$, the state  $q_j$ can reach a final state by reading $u$ . For each word $u = \sigma_1 \dots \sigma_{k-1}$, and every $i$ in $\{0, \dots, k-1\}$, in the $(i, u)$-copy, the final states consist of states $(q_0, \dots, q_k, i, u)$ where $q_i$ is final, and for all $j \neq i$, the state  $q_j$ can reach a final state in the $Q_j$ component by reading $\sigma_{(i-j) \bmod k} \dots \sigma_{k-1}$, i.e., the $(j-i) \bmod k$ length suffix of $u$.  Visually, the brick jutting out the furthest in the brick diagram is the $i$th one, and the rest of the components need to catch up according to the suffix of $u$ they haven't managed to update yet (see \Cref{fig:leapfrogrun}). 

\subsection{Correctness of $\A_{\text{leapfrog}}$}
We will now prove that $\A_{\text{leapfrog}}$ accepts the intersection of the languages of $\A_0, \dots, \A_{k-1}$, demonstrate bounds on its size, and the time needed to construct it. 

\begin{proof}
    Suppose a word $w = \sigma_0\dots\sigma_{t-1}$ is accepted by $\A_{\text{leapfrog}}$.     If $w$ is less than $k-1$ letters long, its run ends in the unique $w$-labelled copy. Since it is accepted, it must have reached a final state in the $w$-copy, which ensures that the $Q_{|w|}$ component is a final state, and all the $Q_j$ for $j < |w|$ components can process the last $|w|-j$ letters of $w$ and reach a final state, and all the $Q_j$ components for $j > |w|$, (which must be starting states since those components have not updated yet), by definition must be able to read $w$ and reach a final state. 
    
    If $w$ has at least $k-1$ letters, we reason as follows. Given a labelled path, let a \emph{path with $k$-skips} starting at $i$ be a sequence of vertices and labels that starts at some index $i < k$, and only visits the $i \bmod k$ positions of the path, and reads only the $i\bmod k$ positions' labels. 
    
    For each run $\pi$ of $w$ on $\A_{\text{leapfrog}}$, the $Q_i$ component of the the restriction $\pi {\upharpoonright_{i \bmod k}}$ is represented as the following path with $k$-skips:  $q_i \xrightarrow[]{\sigma_i} q_{i+k} \xrightarrow[]{\sigma_{i + k}} \dots \xrightarrow[]{\sigma_{i + kx}} q_{i+kx}$ where $x$ is the greatest integer such that $i+kx <t$.  This path with $k$-skips is always the path with $k$-skips of a corresponding run $\rho$ of $\A_i$ on $w$, starting at $i$. Hence, if $\A_{\text{leapfrog}}$ can reach $(q_0, q_1, \dots , q_k, t \bmod  k, u)$ from the initial state on reading $w$, then each component $\A_i$  can start at its initial state, read $\sigma_0 \dots \sigma_{i + kx}$ and reach $q_i$. 
    
    Suppose $w= \sigma_0 \dots \sigma_{t-1}$ was accepted by \frogprod. If on reading $w$, \frogprod reached the final state $(q_0, q_1, \dots , q_k, t \bmod  k, u)$, then in each component $Q_i$ where $i \neq t \bmod k$, one could start at $q_i$, read the last $(t-i) \bmod k$ letters  and reach a final state of $\A_i$, i.e., one has to be able to read $\sigma_{i + kx + 1} \dots \sigma_t$ and accept. If $i=t\bmod k$ then $q_i$ itself must have been a final state in $\A_i$. Since in the $Q_i$ component one can read up to $\sigma_{i + kx}$ and reach $q_i$, and from $q_i$ one can read  $\sigma_{i + kx + 1} \dots \sigma_t$ and accept, this means $w$ has an accepting run in $\A_i$. 

    Conversely, if a word  $w= \sigma_0 \dots \sigma_{t-1}$ belongs to the intersection of $\A_0, \dots, \A_{k-1}$, then it has an accepting run $\rho_i$ in each $\A_i$. Using these runs, we construct an accepting run in \frogprod. Every time \frogprod is in a copy that represents $\A_i$, we mimic the path with $k$-skips starting at $i$ of $\rho_i$.  In any given component, say $Q_i$, start by ensuring that when reaching the $\sigma_0 \dots \sigma_i$-copy, the first transition in the $Q_i$ copy performs the $i$ first moves of the accepting run in $\A_i$, and thereafter whenever in a $(i, u)$-labelled copy update $Q_i$ component by the next $k$ moves of the accepting run of $\A_i$ on $w$. Such a run in $\A_{\text{leapfrog}}$ will necessarily end in a final state in the $Q_{t \bmod k}$ component, and in all other components $Q_i$ the state one reaches must be able to read the $t-i \bmod k$ suffix and reach a final state, since there were exactly $t - i \bmod k$ updates that $\A_i$ did not have time to perform, and $\A_i$ has an accepting run on $w$. Hence, $\A_{\text{leapfrog}}$ recognises the intersection $\bigcap_{i=0}^{k-1} \Lang(\A_i)$. 

    The number of states we use is no more than $k$ copies of $\prod_{i = 0}^{k-1} Q_i$ for every word in $\Sigma^{<k}$, i.e., at most $2k \ell^{k-1} \cdot n^k$. Like the catch-up product, for every $w$ of length up to $k$, for every $i$ in $\{0, \dots, k-1\}$, for every element of $m_{\leq k}$ in $\A_i$ and for every $(k-1)$-tuple of states not in $Q_i$, we have exactly one transition. Hence, the number of transitions is at most $2k \ell^k \cdot m_{\leq k } \hspace{0.25em}n^{k-1}$. 

    The time taken for constructing this automaton depends on the number of states and transitions, $\BigO(k \ell^k \cdot m_{\leq k } \hspace{0.25em}n^{k-1})$, and the complexity of constructing the $u$-reachability relations ($\Delta_i^{u}$ for all $i \in [k]$ and $u \in \Sigma^{\leq k}$) which takes $\BigO(k^2 \ell^k n^{\omega})$. Overall this takes at most $\BigO( k \ell^k \cdot m_{\leq k } \hspace{0.25em}n^{k-1} +k^2 \ell^k n^{\omega} )$ time. 
\end{proof}

%% file: relationsatisfactionproof.tex
\section{Controlled Shuffling}
\label{app:controlled}
In \Cref{sec:sparser}, we start by using the idea that the $k$-stutter of the intersection of $k$ regular languages is the same as the interleaving of the languages intersected with $(a^k + b^k)^*$. This is because the $k$-tape automaton that realises the equality relation between regular languages is very similar to an NFA that recognises $(a^k + b^k)^*$. In the nodding product, we ensure that the component automata move one at a time, to avoid pairing up transitions. In order to ensure they all read the same word, the nodding product mimics the structure of the equality $k$-tape automaton, intuitively letting it control which letters are read by each component. 

In this section, we generalise this observation to allow checking whether $k$ regular languages satisfy any rational relation.

\subsection{Fine-grained complexity of NFA $k$-RS}

\paragraph{Rational relations and $k$-tape automata} We begin by defining the automaton model we call $k$-tape automata, as the Rabin and Scott model for one-way $k$-tape automata; see~\cite[Chap.~4, sect.~1.6]{Sakarovitch_2009} and~\cite{RabScott}. 

\begin{definition}
A \emph{$k$-tape automaton} $\mathcal{C}$ is a tuple $(Q, \Sigma, \Delta, s, F )$ consisting of a finite set of states $Q$, a finite alphabet $\Sigma$, an initial state $s \in Q$, a set of final states $F \subseteq Q$ and a transition relation $\Delta \subseteq Q \times (\Sigma \times [k]) \times Q$.  
\end{definition}

 Each transition $(q, (\sigma, i), q') \in Q \times (\Sigma \times [k]) \times Q$ selects the $i$th of the $k$ read-only tapes, consumes one symbol of input from the selected tape, and updates the current state from $q$ to $q'$. At a given cell, if the $i$th tape is chosen, the automaton reads the letter written in the current cell and moves one step forwards on that tape, remaining stationary on all other tapes.
 
We say the initial tape configuration of $\mathcal{C}$ is $(w_0, w_1, \dots , w_{k-1})$ if the contents of the tapes, for all $i \in [k]$, consists of $w_i$ on the $i$th tape, followed by an endmarker. We say a $k$-tuple of words $(w_0, w_1, \dots , w_{k-1})$ is \emph{accepted} by a \emph{$k$-tape automaton} $\mathcal{C}$ if $\mathcal{C}$ has a run that starts at the initial state, with initial tape configuration $(w_0, w_1, \dots , w_{k-1})$, and reaches a final state, with all $k$ reading heads ending at the endmarkers of their respective tapes. 

The set of $k$-tuples a $k$-tape automaton recognises is denoted by $\Lang(\mathcal{C})$. This is also known as the relation $\mathcal{C}$ recognises. The class of relations $k$-tape automata recognise is known as the $k$-ary rational relations~\cite{DBLP:journals/eatcs/Choffrut06}. 
The special case of $2$-tape automata are usually referred to as transducers, the interpretation being that symbols on the first tape are inputs being read, and symbols on the second tape are outputs produced by the transducer. 

\begin{definition} The $k$-Relation Satisfaction problem ($k$-RS) takes as input $k$ NFA \knfa with $n$ states and $m$ transitions each, over an alphabet of size $\ell$, and a $k$-tape automaton $\mathcal{C}$ with $n_{\mathcal{C}}$ states and $m_{\mathcal{C}}$ transitions. NFA \knfa are said to satisfy $\mathcal{C}$ if there exist words $w_0 \in \Lang(\A_0), \dots, w_{k-1} \in \Lang(\A_{k-1})$ such that $(w_0, w_1, \dots, w_{k-1}) \in \Lang(\mathcal{C})$.
\end{definition}

With this definition, clearly NFA $k$-IE is a special case of the more general NFA $k$-RS problem. 
\begin{lemma}\label{lem:EqRelation}
    Given $k$ NFA with at most $n$ states and $m$ transitions each over an alphabet of size $\ell$, NFA $k$-IE reduces to NFA $k$-RS with the same $k$ NFA and an $\BigO(k \ell)$-sized $k$-tape automaton that recognises the equality relation. This reduction takes $\BigO(k(n + m + \ell))$ time. 
\end{lemma}
Moreover, we can reduce NFA $(k+1)$-IE to NFA $k$-RS.
\begin{theorem}\label{thm:kIEtoRS}
    Given $k+1$ NFA with at most $n$ states and $m$ transitions each, over an alphabet of size $\ell$, NFA $(k+1)$-IE reduces to NFA $k$-RS for $k$ NFA with at most $n$ states and $m$ transitions over an $\BigO(\ell)$-sized alphabet, as well as a $k$-tape automaton at most $k \ell \cdot n$ states and $m + (k-1)\ell \cdot n$ transitions. This reduction takes $\BigO(m + k\ell \cdot n)$ time. 
\end{theorem}

It turns out NFA $k$-RS also reduces to NFA $(k+1)$-IE, by modifying the $k$ input NFA and the $k$-tape automaton $\mathcal{C}$ into appropriate NFA. 

\begin{theorem}\label{thm:kRStoIE}
    Given $k$ NFA with at most $n$ states and $m$ transitions each and a $k$-tape automaton with $n_{\mathcal{C}}$ states and $m_{\mathcal{C}}$ transitions, over an alphabet of size $\ell$, NFA $k$-RS reduces to NFA $(k+1)$-IE for  $k$ NFA with at most $\BigO(n)$ states and $\BigO(m)$ transitions, as well as a new NFA with $\BigO(n_{\mathcal{C}})$ states and $\BigO(m_{\mathcal{C}})$ transitions, over an $\BigO(k\ell)$-sized alphabet. This reduction takes $\BigO(k^2(n + m + l))$ time. 
\end{theorem}

Hence, many of our results on NFA $k$-IE transfer to the NFA $k$-RS problem immediately. 

\begin{corollary} The following statements hold:

\begin{enumerate}
    \item         Given $k$ NFA with at most $n$ states and $m$ transitions each and a $k$-tape automaton with $n_{\mathcal{C}}$ states and $m_{\mathcal{C}}$ transitions, NFA $k$-RS has an {$\BigO(k \cdot n^{k-1} (n_{\mathcal{C}} m + nm_{\mathcal{C}}))$}-time algorithm. 
\item For all $\varepsilon > 0$, if there exists a $\BigO((k \cdot n^{k}n_c)^{1-\varepsilon})$ time algorithm for NFA $k$-RS, then SETH is false. 
\item     NFA $k$-RS belongs to {$\mathsf{NTIME}(k \ell \cdot n^{k + \omega-2 }n_{\mathcal{C}}) \cap \mathsf{coNTIME}(k \cdot n^k n_{\mathcal{C}})$}
\item If there is a fine-grained reduction such that, given SETH, NFA $k$-RS cannot be solved faster than $k \ell  \cdot n^{k + \omega-2+\gamma }n_{\mathcal{C}}$ for some $\gamma > 0$, then NSETH is false. 

\end{enumerate}

\end{corollary}

\subsection{Proof of \Cref{thm:kIEtoRS} and \Cref{thm:kRStoIE}: from NFA $(k+1)$-IE to NFA $k$-RS and back}

\subsubsection*{Reducing NFA $(k+1)$-IE to NFA $k$-RS}
As previously mentioned, the NFA $k$-IE problem can be seen as a special case of the NFA $k$-RS problem, with a fixed $k$-tape automaton for any given $k$. 
But we can go further and reduce NFA $(k+1)$-IE by keeping the first $k$ NFA and using the $k$-tape automaton to capture both the equality relation and the last NFA's language. 

\begin{proof}[Proof  of \Cref{thm:kIEtoRS} (Sketch)]
Given NFA $\A_0, \A_1, \dots , \A_k$, we construct an NFA $k$-RS instance with the $k$ NFA \knfa and a $k$-tape automaton $\mathcal{C}$ that we construct using $\A_k$. For all $i \in [k+1]$, let $\A_i = (Q_i,\Sigma, \Delta_i, s_i, F_i )$. $$\mathcal{C} = (Q_k \cup Q_k\times \Sigma \times \{1, \dots , k-1\}, \Sigma, \Delta_{\mathcal{C}}, s_k, F_k)$$

The set of states of $\mathcal{C}$ consists of a base copy of $Q_k$, as well as $(k-1)\ell$ copies of $Q_k$, each labelled by a letter $\sigma \in \Sigma$ and an index between $1$ and $k-1$. The initial and final states are the same as those of $\A_k$, in the base copy. The transition relation $\Delta_{\mathcal{C}}$ consists of transitions: \begin{enumerate}
    \item $(q, (\sigma,0), (q, \sigma, 1)) \text{ for all } q \in Q_k \text{ and } \sigma \in \Sigma.$ 
    \item $((q, \sigma, i), (\sigma,i), (q, \sigma, i+1)) \text{ for all } q \in Q_k,\, i \in \{1, \dots , k-2\} \text{ and } \sigma \in \Sigma $.
    \item $((q, \sigma, k-1), (\sigma,k-1), q') \text{ for all } (q, \sigma, q') \in \Delta_k$.
\end{enumerate}

Suppose $\A_0, \A_1, \dots , \A_k$ have a word $w$ in their intersection. Then, by definition, $w \in \Lang(\A_i)$ for each $i \in [k]$, but moreover since $w \in \Lang(\A_k)$, there is an accepting run $\rho$ from $s_k$ to $f \in F_k$ in $\A_k$ labelled by $w$. There is a corresponding run in $\mathcal{C}$, where each transition $q \xrightarrow[]{\sigma} q'$ in $\rho$ is replaced by the path $q \xrightarrow[]{(\sigma, 0)} (q, \sigma, 1) \xrightarrow[]{(\sigma, 1)}  \dots \xrightarrow[]{(\sigma, k-2)} (q, \sigma, k-1)\xrightarrow[]{(\sigma, k-1)} q'$. This implies that $(w, w, \dots , w)$ is accepted by $\mathcal{C}$, and hence $\A_0, \dots , \A_{k-1}$ satisfy the relation $\Lang(\mathcal{C})$.                             

Conversely, suppose there exists a $k$-tuple $(w_0, \dots , w_{k-1}) \in \Lang(\mathcal{C})$ such that $w_i \in \Lang(\A_i)$ for all $i \in [k]$. Note that $\mathcal{C}$ only accepts $k$-tuples where all the elements are the same, so there is some word $w$ such that $w_i = w$ for all $i\in [k]$. Moreover, the $k-1 \bmod k$ projection of the accepting run shows that the content of the last tape labels a path from $s_k$ to some $f \in F_k$, in $\A_k$, so $w \in \Lang(\A_k)$ as well. Hence, the languages of $\A_0, \dots , \A_k$ have nonempty intersection, since $w$ is accepted by all of them. 
\end{proof}

\subsubsection*{Reducing $k$-RS NFA to NFA  $(k+1)$-IE}
Interestingly, given an instance of NFA $k$-RS, we can reduce it to deciding NFA $(k+1)$-IE as well. 

Given $k$ NFA \knfa and a $k$-tape automaton $\mathcal{C}$, we reduce the NFA $k$-RS problem to the NFA $(k+1)$-IE problem for NFA $\mathcal{B}_0 , \dots , \mathcal{B}_{k}$.

The idea of the reduction is transcribing the moves the $k$-tape automaton $\mathcal{C}$ makes on its tapes,  with the additional constraint that the contents of the $i$th tape must belong to $\Lang(\A_i)$ for each $i\in [k]$. To do so, we modify the $k$ input NFA and the $k$-tape automaton $\mathcal{C}$ so that their intersection captures a shuffled (not interleaved) version of the $k$ components of the relation $(\Lang(\A_0) \times \dots \times \Lang(\A_{k-1})) \cap \Lang(\mathcal{C})$.  

The NFA we construct will use the alphabet $\Gamma = \Sigma \times [k]$. 
Intuitively,  $(\sigma, i)$ corresponds to $\A_i$ reading $\sigma$, as well as to $\mathcal{C}$ reading $\sigma$ in the current cell of the $i$th tape. For each $i\in [k]$, let $\gamma_i: \Sigma^* \to \Gamma^*$ be the homomorphism that maps each letter $\sigma \in \Sigma$ to $(\sigma, i)$. For each $i\in [k]$, the function $\gamma_i$ extends to languages in the natural way. 

In order to reduce the NFA $k$-RS problem to the NFA $(k+1)$-IE problem, we would like the new NFA $\mathcal{B}_0, \dots \mathcal{B}_k$ to have the following property: $$\bigcap_{i \in [k+1]} \Lang(\mathcal{B}_i) = \emptyset \iff (\Lang(\A_0) \times \dots \times \Lang(\A_{k-1})) \cap \Lang(\mathcal{C}) = \emptyset.$$ 

Let the \emph{ordinary shuffle}~\cite{HenshallRS12} of two languages $L_0$ and $L_1$ over an alphabet $\Sigma$ refer to the language: $$
    \{x_0y_0 x_1y_1 \dots x_{n-1}y_{n-1} \mid n\geq 0 \text{ and  for all } i \in [n],\, x_i, y_i \in \Sigma^*,\, x_0x_1\dots x_{n-1} \in L_0,\, \text{and } y_0y_1\dots y_{n-1}\in L_1 \}.
$$ The ordinary shuffle of $k$ languages can be defined analogously. 

\begin{proof}[Proof (Sketch) of \Cref{thm:kRStoIE}]
For all $i \in [k]$, the NFA $\mathcal{B}_i$ has the same state space as $\A_i$, namely $Q_i$, but its transitions are slightly different. 
$$\mathcal{B}_i = (Q_i, \Gamma, \Delta'_i, s_i, F_i )$$
The initial state and final states are the same as those of $\mathcal{A}_i$. The transition relation of $\mathcal{B}_i$ contains: 
\begin{itemize}
    \item for all $j \neq i$, $q \xrightarrow[]{(\sigma, j) }q$ for all $q \in Q_i$ and $\sigma \in \Sigma$ and
    \item $q \xrightarrow[]{(\sigma, i) }q'$ whenever $q \xrightarrow[]{\sigma} q'$ in $ \A_i$.
\end{itemize} 

For all $i \in [k]$, the language the NFA $\Lang(\mathcal{B}_{i})$ is the ordinary shuffle of two languages, $\gamma_i(\Lang(\A_i))$ and $(\Sigma \times ([k]\setminus \{i\}))^*$. Hence, all the letters labelled by $i$ in any word $w \in \Lang(\mathcal{B}_{i})$ spell a word belonging to $\Lang(\A_i)$, with no restriction on the other letters. 

The NFA $\mathcal{B}_{k}$ modifies the original $k$-tape automaton $\mathcal{C}$. $\mathcal{B}_k$ is a one-tape automaton that reads $(\sigma,i)$ whenever $\mathcal{C}$ reads $\sigma$ on the $i$th tape. The NFA $\mathcal{B}_{k}$ is defined as follows: $$\mathcal{B}_{k} = (Q_{\mathcal{C}}, \Gamma, \Delta_{k}, s_{\mathcal{C}}, F_{\mathcal{C}}).$$ 
Once again, the state space, initial state, and final states are the same as $\mathcal{C}$. For each transition $q \xrightarrow[]{(\sigma,i)} q'$ in $\mathcal{C}$ (reading the letter $\sigma$ on the $i$th tape), we have a transition $q \xrightarrow[]{(\sigma, i)} q'$ in $\Delta_{k}$ (reading the letter $(\sigma, i)$). 

The language $\Lang(\mathcal{B}_{k})$ contains a word $w = (\sigma_0, i_0)(\sigma_1, i_1) \dots (\sigma_{p-1}, i_{p-1})$ whenever there exists a tuple $(w_0, \dots , w_{k-1})$ that had an accepting run $t_0 t_1 \dots t_{p-1}$ in $\mathcal{C}$, where for all  $j \in [p]$, $t_{j}$ is a transition that instructs $\mathcal{C}$ to read $\sigma_j$ on the $i_j$th tape. This implies that $w$ belongs to the ordinary shuffle of the $k$ languages $\{\gamma_i(w_i)\} \subseteq \Gamma^*$, $i \in [k]$. 

The intersection $\bigcap_{i \in [k]} \Lang(\mathcal{B}_{i})$ is the ordinary shuffle of the $k$ languages $\gamma_i(\Lang(\A_i))$ for all $i \in [k]$. Moreover, \begin{multline*} \qquad\Lang(\mathcal{B}_{k}) = \{w \mid w \text{ belongs to the ordinary shuffle of  the $k$ languages } \\ \{\gamma_i(w_i)\} \subseteq \Gamma^*
         \text{ for all }i \in [k] \text{, for all }(w_0, \dots , w_{k-1}) \in \Lang(\mathcal{C})\}.
         \end{multline*} Hence, the new NFA $\mathcal{B}_0, \dots , \mathcal{B}_k$ have a nonempty intersection exactly when, for each $i \in [k]$, there are words $w_i \in \Lang(\A_i)$ such that $(w_0, \dots , w_{k-1})$ had an accepting run in the $k$-tape automaton $\mathcal{C}$. 
\end{proof}